\documentclass[runningheads,envcountsame,a4paper]{llncs}
\usepackage{graphicx}
\usepackage{latexsym,amssymb,amsmath}
\usepackage[utf8x]{inputenc}
\usepackage[all]{xy}
\usepackage{gastex}
\usepackage{color}
\usepackage[active]{srcltx}
\usepackage{float}
\usepackage[shortcuts,cyremdash]{extdash}
\numberwithin{theorem}{section} 
\label{definitions}

\usepackage{url}
\urldef{\mailsa}\path|{klima,polak}@math.muni.cz|

\newcommand{\Past}{\operatorname{\mathsf Past}}
\newcommand{\id}{{\mathsf{id}}}
\newcommand{\cl}[1]{{\uparrow\!\!#1}}
\newcommand{\ol}[1]{{\overline{#1}}}
\newcommand{\OH}{{\mathsf H }}
\newcommand{\OS}{{\mathsf{S}}}
\newcommand{\OP}{{\mathsf{P}}}
\newcommand{\OD}{{\mathsf{D}}}
\newcommand{\OI}{{\mathsf{I}}}
\newcommand{\XX}{{\mathbf{X}}}
\newcommand{\YY}{{\mathbf{Y}}}
\font\aa=rsfs10 \def\nL{\mbox{\aa L\,}}
\def\Od{\mathsf d}
\def\Ou{\mathsf u}
\def\Oo{\mathsf o}

\begin{document}

\title{On Varieties of Ordered Automata\thanks{The paper was supported by 
grant GA15-02862S of the Czech Science Foundation.}}
\author{Ond\v rej Kl\'{i}ma and Libor Pol\'ak}
\institute{Department of Mathematics and Statistics, Masaryk University,\\
Kotl\'a\v rsk\'a 2, 611 37 Brno, Czech Republic,
\mailsa\\
\url{http://www.math.muni.cz}}

\maketitle

\begin{abstract}
The Eilenberg correspondence relates varieties of regular languages 
to pseudovarieties of finite monoids. 
Various modifications of this correspondence have been found with more general 
classes of regular languages on one hand and classes of more complex algebraic 
structures on the other hand.
It is also possible to consider classes of automata  instead of algebraic 
structures as a natural counterpart of classes of languages. 
Here we deal with the correspondence relating positive $\mathcal C$\=/varieties of languages 
to positive $\mathcal C$\=/varieties of ordered automata and we present various specific instances 
of this correspondence.
These bring certain well\=/known results from a new perspective and also some new observations.
Moreover, complexity aspects of the membership problem are  
discussed both in the particular examples and in a general setting.  
\end{abstract}


\section{Introduction}\label{Introduction}
Algebraic theory of regular languages is a well\=/established field in the theory of formal languages.
The basic ambition of this theory is to obtain 
effective characterizations of various natural classes of regular languages.
First examples of significant classes of languages, which were effectively characterized 
by properties of syntactic monoids, were the star\=/free languages
by Sch{\"u}tzenberger~\cite{schutz} 
and the piecewise testable languages by Simon~\cite{simon-pw}. 
A~general framework for discovering relationships
between properties of regular languages and properties of monoids
was provided by Eilenberg~\cite{eilenberg}, who established a one\=/to\=/one
correspondence between the so\=/called {\em varieties} of regular languages
and \emph{pseudovarieties} of finite monoids.
Here varieties of languages are classes closed for taking quotients, 
preimages under homomorphisms and Boolean operations.
Thus a membership problem for a given variety of regular languages can be translated to a membership problem
for the corresponding pseudovariety of finite monoids.
An advantage of this approach is that pseudovarieties of monoids are exactly classes
of finite monoids which have an equational description by pseudoidentities -- see Reiterman~\cite{reiterman}.
For a thorough introduction
to that theory we refer to
surveys by Pin~\cite{pin-handbook} and by Straubing and Weil~\cite{straubing-weil-handbook}.

Since not every natural class of languages is closed for taking all mentioned operations,
various generalizations of the notion of varieties of languages have been
studied. 
One possible generalization is the notion of {\em positive varieties} of languages introduced 
by Pin~\cite{pin-positive} -- the classes need not be closed for taking complementation.
Their equational characterization was given by Pin and Weil~\cite{pin-weil}.
Another possibility is to weaken the closure property concerning 
preimages under homomorphisms -- only 
homomorphisms from a certain fixed class $\mathcal C$ are used.
In this way, one can consider $\mathcal C$\=/varieties of regular languages
which were introduced by Straubing~\cite{straubing} 
and whose equational description was presented by Kunc~\cite{michal}.
These two generalizations could be combined 
as suggested by Pin and Straubing~\cite{pin-straubing}. 

In our contribution we do not use syntactic structures at all. 
We consider classes of automata as another
natural counterpart to classes of regular languages. 
In fact, we deal with classes of semiautomata, which are exactly
automata without the specification of initial nor final states.
Characterizing of classes of languages by properties of minimal automata
is quite natural, since usually we assume that an input of a membership problem
for a fixed class of languages is given exactly by the minimal deterministic automaton.
For example, if we want to test whether an input language is piecewise testable,
we do not need to compute its syntactic monoid which could be quite 
large (see Brzozowski and Li~\cite{broz-j-triv}).
Instead of that, we check 
a condition which must be satisfied by its minimal automaton 
and which was also established in~\cite{simon-pw}.
This characterization was used in~\cite{stern} and~\cite{trahtman}
to obtain a polynomial and quadratic algorithms, respectively, 
for testing piecewise testability.
In~\cite{dlt13-klima}, Simon's condition was reformulated and 
the so\=/called {\em confluent acyclic (semi)automata} were defined.
In this setting, this characterization can be viewed as an instance
of Eilenberg type theorem between varieties of languages and varieties of semiautomata.

Moreover,
each minimal automaton is implicitly equipped with an order
in which the final states form an upward closed subset. This leads to a 
notion of ordered automata.
Then positive $\mathcal C$\=/varieties of ordered semiautomata can be defined 
as classes which are closed for taking certain natural closure operations. 
We recall here the general Eilenberg type theorem, namely 
Theorem~\ref{t:eilenberg-ordered},
which states that positive $\mathcal C$\=/varieties of ordered 
semiautomata correspond to positive $\mathcal C$\=/varieties of languages. 

Summarizing, there are three worlds:\\
(L) classes of regular languages,\\
(S) classes of finite monoids,
sometimes enriched by an additional structure like 
the ordered monoids, monoids with distinguished generators, etc.,\\
(A) classes of semiautomata, sometimes ordered semiautomata, etc.

Most variants of Eilenberg correspondence relate (L) and (S),
the relationship between (A) and (S) was studied by Chaubard et al.~\cite{pin-str}, and 
finally
the transitions between (L) and (A) were initiated by {\'E}sik and Ito~\cite{esik-ito}. 
Here
we continue in the last approach, 
to establish Theorem~\ref{t:eilenberg-ordered}. In fact,
this result is a combination of Theorem 5.1 of \cite{pin-straubing} 
(only some hints to a possible proof are given there) and
the main result of \cite{pin-str} relating worlds (S) and (A).
In contrary, in~the present paper, one can find a self\=/contained proof 
which does not go through the classes of monoids.

The paper is structured as follows.
In Sections 2 and 3 we recall the basic notions. In Sections 4 and 5
we study ordered semiautomata and some natural algebraic constructions 
on them.
The next section is devoted to the detailed proof of Theorem~\ref{t:eilenberg-ordered}. 
Section 7 explains how the unordered variant of this result 
can be obtained. Section~8 presents several instances 
of Theorem~\ref{t:eilenberg-ordered} and 
Section~9 discusses membership problem for $\mathcal C$\=/varieties of semiautomata
given by certain type of pseudoidentities.

This paper is a technical report which precedes the short conference paper~\cite{kp-lata} on the topic. 
The final authenticated publication is available online at {\tt https://doi.org/10.1007/978-3-030-13435-8}.

\section{Positive $\mathcal C$\=/Varieties of Languages}

First of all, we recall  basic definitions. 
Let $A^*$ be the set of all words over a finite alphabet $A$. 
We denote by $\lambda$ the empty word.
The set $A^*$ equipped with the operation of concatenation forms a free 
monoid over $A$ with
$\lambda$ being a neutral element.
A {\it language} over alphabet $A$ is a subset of $A^*$. 
Note that all languages which are considered in the paper are regular.
For a language $L\subseteq A^*$ and
a pair of words $u,v\in A^*$, we denote by $u^{-1}Lv^{-1}$ the 
{\it quotient} of $L$ by these words, i.e.
the set $u^{-1}Lv^{-1}=\{\, w\in A^* \mid uwv\in L\, \}$.
In particular, a {\em left quotient} is defined by 
$u^{-1} L = \{\,w\in A^* \mid uw\in L\,\}$ and a
{\em right} one is defined by $L v^{-1} = \{\,w\in A^* \mid wv\in L\,\}$.

For the propose of this paper, following Straubing~\cite{straubing},
the {\em category of homomorphisms} $\mathcal C$
is a category where objects are all free monoids over non-empty finite alphabets 
and morphisms are certain monoid homomorphisms
among them. If the sets $A$ and $B$ are clear from the context,
we write briefly $f\in \mathcal C$
instead of $f\in \mathcal C(A^*,B^*)$.
This ``categorical'' definition means
that $\mathcal C$
satisfies the following properties: 
\begin{itemize}
\item For each finite alphabet $A$, the identity mapping 
$\id_A : A^* \rightarrow A^*$ belongs to $\mathcal C$.
\item  
If $f: B^* \rightarrow A^*$ and 
$g: C^* \rightarrow B^*$ belong to $\mathcal C$, then their composition 
$gf : C^* \rightarrow A^*$ is also in $\mathcal C$.
\end{itemize}

If $f: B^* \rightarrow A^*$ is a homomorphism and $L\subseteq A^*$, then 
by the {\em preimage} of $L$ in the homomorphism $f$ is meant the set 
$f^{-1}(L) =\{\, v \in B^* \mid f(v)\in L\, \}$.

\begin{definition}\label{d:positive-C-varieties}
Let $\mathcal C$ be a category of homomorphisms.
A {\em positive $\mathcal C$\=/variety of languages} $\mathcal{V}$
associates to every non\=/empty finite alphabet $A$ a
class $\mathcal{V}(A)$ of regular languages over $A$ in such a way that
\begin{itemize}
\item
$\mathcal{V}(A)$ is closed under unions and intersections of finite families,
\item $\mathcal{V}(A)$ is closed under
quotients, i.e.
$$L\in \mathcal {V}(A),\ u,v\in A^*\quad 
\text{implies} \quad
u^{-1}Lv^{-1}\in\mathcal {V}(A)\, ,$$
\item  $\mathcal{V}$ is closed under preimages in 
morphisms of $\mathcal C$, i.e.
$$f:B^*\rightarrow A^*,\, f\in \mathcal C,\,  L\in \mathcal {V}(A)\quad  
\text{implies}\quad  
f^{-1}(L)\in \mathcal{V}(B)\, .$$
\end{itemize}
\end{definition}

Note that the first condition in Definition~\ref{d:positive-C-varieties} ensures that 
the languages $\emptyset$ and $A^*$ belong to $\mathcal{V}(A)$ for every 
alphabet $A$:
$\emptyset$ is the union of the empty system and 
$A^*$ is  the intersection of the 
empty system. 
In other words, the first condition can be equivalently formulated as
$\emptyset, A^*\in \mathcal{V}(A)$ and $\mathcal{V}(A)$ is closed under binary unions and intersections. 
In particular, all $\mathcal{V}(A)$'s are nonempty.

If $\mathcal C$ consists of all homomorphisms we get exactly 
the notion of the 
positive varieties of languages. When adding ``each $\mathcal{V}(A)$ is closed under
complements'', we get exactly the notion of the  $\mathcal C$\=/\emph{variety}
of languages.


\section{The Canonical DFA}

In this section we fix basic terminology concerning finite automata.
First of all, note that all considered automata in the paper are 
deterministic, complete, finite and over finite alphabets.  
Moreover, we use the term {\it semiautomaton} when the initial and final
states are not explicitly given.

A {\em deterministic finite automaton} (DFA) over the alphabet $A$ is
a five\=/tuple $\mathcal A = (Q,A,\cdot,i,F)$, where
$Q$ is a non\=/empty set of {\em states},
$\cdot : Q\times A \rightarrow Q$ is a complete {\em transition function}, 
$i\in Q$ is the {\em initial} state and
$F\subseteq Q$ is the set of {\em final} states.
The transition function can be extended to a mapping 
$\cdot : Q\times A^* \rightarrow Q$ by
$q\cdot \lambda = q,\ q\cdot (ua)= (q\cdot u)\cdot a$,
for every $q\in Q,\,u\in A^*,\,a\in A$.
The automaton $\mathcal A$ {\it accepts} a word $u\in A^*$
if and only if $i\cdot u\in F$ and 
the language {\em recognized} by the
automaton $\mathcal A$ is 
$\nL_{\mathcal A}=\{\,u\in A^* \mid i\cdot u\in F\,\}$.
More generally, for $q\in Q$, we denote
$\nL_{\mathcal A,q}=\{\,u\in A^* \mid q\cdot u\in F\,\}$.
For a fixed $\mathcal A$, 
we denote this language simply by $\nL_q$.

We recall the construction of the minimal automaton of a regular language
which was introduced by Brzozowski~\cite{broz-minimal}. Since 
this automaton is uniquely determined and it
plays a central role in our paper,
we use the adjective ``canonical'' for it.

\begin{definition}\label{d:canonical-dfa}
The {\em canonical deterministic automaton} of a regular language $L$ is 
$\mathcal D_L =(D_L,A,\cdot,L,F_L)$, where
$D_L =\{\, u^{-1} L \mid u\in A^*\,\}$,
$q\cdot a =  a^{-1}q$, for each $q\in D_L$, $a\in A$,
and  $F_L=\{\, q\in D_L \mid  \lambda\in q\,\} $. 
\end{definition}

A part of Brzozowski's result is the correctness of the previous definition, because 
one needs to show that $\mathcal D_L$ is really a finite deterministic automaton. 
The minimality of $\mathcal D_L$ 
can be obtained as a consequence of the following lemma.
Since the result will be modified later in the paper, 
the proof of the following lemma is also presented here.

\begin{lemma}[\cite{broz-minimal}]\label{l:canonical-dfa}
Let $L$ be a regular language with the canonical automaton $\mathcal D_L$
and let $\mathcal A=(Q,A,\cdot,i,F)$ be an arbitrary DFA with 
$\nL_{\mathcal A}=L$. Then the following holds:

(i) For each  $q\in D_L$, we have that $\nL_{\mathcal D_L, q}=q$.

(ii) For each $u\in A^*$, we have that $\nL_{\mathcal A, i\cdot u}=u^{-1}L$.

(iii) The rule $\varphi : i\cdot u \mapsto u^{-1}L$, 
for every $u\in A^*$, correctly defines
a surjective mapping from $Q'=\{\,i\cdot u\mid u\in A^*\,\}$ onto $D_L$
satisfying $\varphi((i\cdot u)\cdot a)= (\varphi(i\cdot u))\cdot a$, 
for every $u\in A^*$, $a\in A$.
\end{lemma}
\begin{proof} (i)
Let $q$ be a state of $\mathcal D_L$, i.e. $q=u^{-1}L$ for some $u\in A^*$.
Then $v\in \nL_{\mathcal D_L, q}$ if and only if $\lambda\in (u^{-1}L)\cdot v=(uv)^{-1}L$,
which is equivalent to $uv\in L$ and also to $v\in u^{-1}L=q$.

(ii) Let $u\in A^*$. Then, 
for every $v\in A^*$,
we have the following chain of equivalent formulas:
$$v\in\nL_{\mathcal A, i\cdot u},\ (i\cdot u)\cdot v\in F,\ 
\ uv\in L, \ \text{and} \
\ v\in u^{-1}L\, .$$
 
(iii) 
The correctness of the definition of $\varphi$ follows from (ii) and the
surjectivity of $\varphi$ is clear.
Moreover, $\varphi((i\cdot u)\cdot a)= \varphi(i\cdot ua)=(ua)^{-1}L
=u^{-1}L\cdot a=\varphi(i\cdot u)\cdot a$.
\qed\end{proof}


\section{Ordered Automata}\label{sec:ordered-automata}

First,
we recall some basic terminology from the theory of ordered sets.
By an {\it ordered set} we mean a set $M$ equipped with an 
{\it order} $\le$, i.e. by a reflexive, antisymmetric and transitive 
relation. A subset $X$ is called {\it upward closed}
if, for every pair of elements $x,y\in M$, 
the following property holds: $x\le y,\, x\in X$ implies $y\in X$.
For every subset $X$, we denote by $\cl X$ the smallest upward closed subset containing the subset $X$,
i.e. $\cl X=\{\,m \in M \mid \exists\, x\in X : x\le m\,\}$. 
In particular, for $x\in M$, we write $\cl x$ instead of $\cl\{x\}$.
A mapping $f: M \rightarrow N$ between two ordered sets $(M,\le)$ and $(N,\le)$
is called {\it isotone} if, for every pair of elements $x,y\in M$, we have 
that $x\le y$ implies $f(x)\le f(y)$.

States of the canonical automaton $\mathcal D_L$ are languages, 
and therefore they are ordered naturally by the set\=/theoretical inclusion. 
The action by each letter $a\in A$ is an isotone mapping:  
for each pair of states $p,q$ such that 
$p\subseteq q$, we have  $p\cdot a= a^{-1}p \subseteq a^{-1}q=q\cdot a $.
Moreover, the set $F_L$ of all final states is an upward closed subset 
with respect to $\subseteq$.
These observations motivate the following definition.

\begin{definition}\label{d:ordered-automaton}
An {\em ordered automaton} over the alphabet $A$ is
a six\=/tuple $\mathcal O = (Q,A,\cdot,\le,i,F)$, where
\begin{itemize}
\item $(Q,A,\cdot,i,F)$
is a usual DFA;
\item $\le $ is an order on the set $Q$;
\item the action by every letter $a\in A$ is an isotone mapping from 
the ordered set $(Q,\le)$ to itself;
\item $F$ is  an upward closed subset of $Q$ with respect to $\le$.
\end{itemize}
\end{definition}

The definitions of the acceptance and the recognition are the 
same as in the case of DFA's.  
Since a composition of isotone mappings is isotone, 
it follows from  Definition~\ref{d:ordered-automaton} that the action by 
every word $u\in A^*$ is an isotone mapping
from the ordered set of states into itself.

Moreover, the ordered semiautomaton $\mathcal Q=(Q,A,\cdot,\leq)$
{\it accepts} the language $L\subseteq A^*$ if we
can complete $\mathcal Q$ to an ordered
automaton $\mathcal O = (Q,A,\cdot,\le,i,F)$ such that $L=\nL_{\mathcal O}$.

The following result states that Brzozowski's construction gives 
the minimal ordered automata $\mathcal O_L =(D_L,A,\cdot,\subseteq,L,F_L)$.

\begin{lemma}\label{l:canonical-odfa}
Let $\mathcal O = (Q,A,\cdot,\le,i,F)$ be an ordered automaton 
recognizing
the language $L=\nL_{\mathcal O}$.
Then, for states $p\le q$,
we have $\nL_p\subseteq \nL_q$.
Moreover, the mapping $\varphi$ from Lemma \ref{l:canonical-dfa} 
is an isotone one onto the canonical ordered automaton
$\mathcal O_L =(D_L,A,\cdot,\subseteq,L,F_L)$.
\end{lemma}

\begin{proof}
Let $p\le q$ hold in the given ordered automaton $\mathcal O$.
 If $w\in\nL_p$, then $p\cdot w\in F$. Now $p\cdot w \le q\cdot w$ 
 implies $q\cdot w \in F$ and therefore  $w\in \nL_q$.
Moreover, having $u,v\in A^*$ such that $p=i\cdot u,\,q=i\cdot v$, we get
 $\nL_{i\cdot u}\subseteq \nL_{i\cdot v}$. By Lemma~\ref{l:canonical-dfa},
 it means that $u^{-1}L\subseteq v^{-1}L$, i.e. $\varphi(i\cdot u)
 \subseteq \varphi(i\cdot v).$
\qed\end{proof}

When relating languages with algebraic structures (not our task here), 
the following property of the minimal/canonical ordered automaton is a crucial one.

\begin{lemma}[Pin~{\cite[Section 3]{pin-handbook}}]
The transition monoid of the minimal automaton of a regular language
$L$ is isomorphic to the syntactic monoid of~$L$.
Similarly, the ordered transition monoid of the minimal ordered automaton of 
$L$ is isomorphic
to the syntactic ordered monoid of $L$.
\end{lemma}

The next lemma clarifies how 
the quotients of a language can be obtained changing
the initial and the final states appropriately.

\begin{lemma}\label{l:quotients}
Let $\mathcal O=(Q,A,\cdot,\leq,i,F)$ be an ordered automaton recognizing
the language $L=\nL_{\mathcal O}$.
Let $u,v\in A^*$. Then

(i) $u^{-1}L=\nL_{\mathcal B}$ where  
$\mathcal B=(Q,A,\cdot,\leq,i\cdot u,F)$,

(ii)
$Lv^{-1}=\nL_{\mathcal C}$ where  
$\mathcal C=(Q,A,\cdot,\leq,i,F_v) \text{ and } 
F_v=\{\,q\in Q\mid q\cdot v\in F\,\}$.
\end{lemma}

\begin{proof}
(i)
It follows from Lemma~\ref{l:canonical-dfa} (ii).

(ii)
To show that $\mathcal C$ is an ordered semiautomaton, 
we need to prove that $F_v$ is upward closed.
Let $p\in F_v$ and $p\le q\in Q$.
From $p\in F_v$ we have $p\cdot v\in F$ and from $p\le q$ 
we obtain $p\cdot v\le q\cdot v$. Since $F$ is upward closed
we get $q\cdot v\in F$, which implies $q\in F_v$.

Now, for every $w\in A^*$, 
the following is a chain of equivalent statements:
$$w\in Lv^{-1},\ wv\in L,\ (i\cdot w)\cdot v\in F,\ 
i\cdot w\in F_v,\ \text{and}\ w \in \nL_{\mathcal C}\,.$$
Thus we proved the equality  $Lv^{-1}=\nL_{\mathcal C}$. 
\qed\end{proof}

The next result characterizes languages which are recognized
by changing the final states in the canonical ordered automaton.

\begin{lemma}\label{l:quotients-in-canonical}
Let $H\subseteq D_L$ be upward closed. Then
$(D_L,A,\cdot,\subseteq, L,H)$ recognizes a language which can be 
expressed as a finite union of finite intersections
of languages of the form $Lv^{-1}$.
\end{lemma}

\begin{proof}
For an arbitrary upward closed subset $X\subseteq D_L$,
we define $\Past(X)=\{\,w\in A^* \mid L \cdot w \in X\,\}$.
Since $\Past(X)=\bigcup_{p\in X}\Past(\cl p)$, it is enough to prove
that, for each $p\in D_L$, the set $\Past(\cl p)$ can be expressed as 
a finite intersection
of right quotients of the language $L$.

Let $p$ be an arbitrary state in $D_L$.
For each $q\in D_L$ such that $p\not \subseteq q$, we have 
$p=L_p\not\subseteq L_q=q$. This means
that there is $v_q\in A^*$ with the property 
$v_q \in L_p \setminus L_q$. Equivalently, 
$p\cdot v_q\in F_L$ and $q\cdot v_q\not\in F_L$.
Now if $w\in \Past(\cl p)$ then $L\cdot w \supseteq p$. Therefore 
$(L \cdot w) \cdot v_q \supseteq p\cdot v_q\in F_L$, 
from which we get $wv_q \in L$, i.e. $ w\in L v_q^{-1}$. 
We have showed that $\Past(\cl p) \subseteq L v_q^{-1}$.

Now we claim that 
$$\Past(\cl p)=\bigcap_{p\not\subseteq q} L v_q^{-1} \, .$$
We already saw the ``$\subseteq$''\=/part. 
To prove the opposite inclusion, 
let $w$ be an arbitrary word from $\in \bigcap_{p\not\subseteq q} L v_q^{-1}$.
Fixing $q$ for a moment, we see that $wv_q\in L$, i.e.  
$L \cdot wv_q \in F_L$.
In the case of $q=L\cdot w$
we would have $q\cdot v_q=(L \cdot w)\cdot v_q\in F_L$ -- a contradiction.
Hence
$L\cdot w\not= q$, and this holds for each $q\not\supseteq p$. 
Therefore $L\cdot w\supseteq p$ and we deduce that $w\in \Past(\cl p)$.
\qed\end{proof}


There is a natural question how the minimal ordered (semi)automaton 
can be computed from a given automaton. 

\begin{proposition}
There exists an algorithm which  computes, for a given  
automaton $\mathcal A=(Q,A,\cdot, i, F)$,
the minimal ordered automaton
of the language $\nL(\mathcal A)$.
\end{proposition}
\begin{proof}
Our construction is based on Hopcroft minimization algorithm for DFA's.
We may assume that all states of $\mathcal A$
are reachable from the initial state $i$.
Let $R=(Q\times F) \cup ((Q\setminus F) \times Q)$.
Then we construct the relation $\overline R$ from $R$ by removing 
unsuitable pairs of states step by step.
At first, we put $R_1=R$. Then for each integer $k$, 
if we find $(p,q)\in R_k$ and a letter $a\in A$
such that $(p\cdot a,q\cdot a) \not\in R_k$, 
then we remove $(p,q)$ from the current relation $R_k$, that is, 
we put $R_{k+1}=R_k \setminus \{(p,q)\}$.
This construction stops after, say, $m$ steps. 
So, $R_{m+1}=\overline R$ satisfies $(p,q)\in \overline R \implies 
(p\cdot a,q\cdot a) \in \overline R$, for every $p,q\in Q$ and $a\in A$. 
Now, we observe that, $(p,q)\in \overline R$ if and only if, for every $u\in A^*$, 
$(p\cdot u, q\cdot u) \in R$. 
The condition can be equivalently written  as
\begin{equation}
    \label{eq:hopcroft}
(p,q)\in \overline R \ \text{ if and only if }
\  (\ \forall\, u\in A^* \ :\ p\cdot u \in F\ \implies\ q\cdot u \in F\ )\, .
\end{equation}    
It follows that  $\overline R$ 
is a quasiorder on $Q$ and we can consider the 
corresponding equivalence relation 
$\rho=\{\,(p,q) \mid (p,q)\in \overline R , (q,p) \in \overline R\,\}$ on the set $Q$.
Then the quotient set $Q/\rho =\{\,[q]_\rho \mid q\in Q\,\}$ has a structure of the automaton:
the rule $[q]_\rho\cdot_\rho a =[q\cdot a]_\rho$, 
for each $q\in Q$ and $a\in A$, defines correctly 
actions by letters using~(\ref{eq:hopcroft}). 
Furthermore, the relation $\le$ on $Q/\rho$ defined by the rule
$[p]_\rho \le [q]_\rho$ iff $(p,q)\in \overline R$, is an order on $Q/\rho$ 
compatible with actions by letters.
So, $\mathcal A_\rho=(Q/\rho, A, \cdot_\rho,\le, [i]_\rho, F_\rho)$, where 
$F_\rho=\{\,[f]_\rho \mid f\in F\,\}$,
is an ordered automaton recognizing $\nL(\mathcal A)$. 
Moreover, if there are two states $[p]_\rho,[q]_\rho \in Q/\rho$ 
such that $\nL(\mathcal A_\rho,p)=\nL(\mathcal A_\rho,q)$, then
$(p,q)\in \rho$.
Thus, the ordered automaton $\mathcal A_\rho$ is isomorphic to 
the minimal ordered automaton of the language $\nL(\mathcal A)$.
\qed\end{proof}

Note also that the classical power\=/set construction makes from 
a nondeterministic automaton
an ordered deterministic automaton which is ordered by the set\=/theoretical inclusion. 
Thus, for the purpose 
of a construction of the minimal ordered automaton, 
one may also use Brzozowski's minimization algorithm  
using power\=/set construction for the reverse of the given language. 


\section{Algebraic Constructions on Ordered Semiautomata}
\label{sec:constructions}

To get an Eilenberg correspondence between classes of languages and the classes
of semiautomata we need an appropriate definition of a variety of semiautomata. 
The notion of variety of semiautomata would be given in terms of closure properties
with respect to certain constructions on semiautomata.

Positive $\mathcal C$\=/varieties of languages are closed under quotients, 
therefore the choice of an initial state and final states in ordered 
automata can be left free due to  Lemma~\ref{l:quotients}.

If an ordered automaton 
$\mathcal O=(Q,A,\cdot,\le, i, F)$
is given, then we denote by $\ol{\mathcal O}$ the corresponding ordered 
semiautomaton $(Q,A,\cdot,\le)$.
In particular, for the canonical ordered automaton
$\mathcal{D}_L=(D_L,A,\cdot,\subseteq,L,F)$ of the language $L$, 
we have $\ol{\mathcal{D}_L}=(D_L,A,\cdot,\subseteq)$.

Since positive $\mathcal C$\=/varieties of languages are closed under 
taking finite unions and intersections, 
we include the closedness with respect to
direct products of ordered semiautomata.

\begin{definition}\label{d:product}
Let $n\ge 1$ be a natural number.
Let $\mathcal O_j=(Q_j,A,\cdot_j,\le_j)$ be an ordered semiautomaton 
for $j=1,\dots,n$. 
We define the ordered semiautomaton 
$\mathcal O_1\times\dots\times \mathcal O_n=
(Q_1\times\dots\times Q_n,A,\cdot,\le)$ as follows:
\begin{itemize}
\item for each $a\in A$, we put
$(q_1,\dots,q_n) \cdot a = (q_1\cdot_1 a,\dots,q_n\cdot_n a)$ and
\item we have 
$(p_1,\dots,p_n) \le (q_1,\dots,q_n)$ if and only if, for each 
$j=1,\dots,n$,
the inequality $p_j\le_j q_j$ is valid.
\end{itemize}
The ordered semiautomaton $\mathcal O_1\times\dots\times \mathcal O_n$
is called a {\em product}
of the ordered semiautomata $\mathcal O_1,\dots,\mathcal O_n$.
\end{definition}

We would like to know, which languages are recognized by a product 
of ordered semiautomata.

\begin{lemma} \label{l:product-arbitrary-subset}
Let the ordered semiautomaton $\mathcal O$ be the product
of the ordered semiautomata $\mathcal O_1,\dots,\mathcal O_n$.
Then the following holds:

(i) If, for each $j=1,\dots,n$, the language 
$L_j$ is recognized by $\mathcal O_j$, then both $L_1\cap\dots\cap L_n$
and $L_1\cup\dots\cup L_n$ are recognized by $\mathcal O$.

(ii) If the language $L$ is recognized by $\mathcal O$, then $L$
is a finite union of finite intersections of languages recognized by
$\mathcal O_1,\dots,\mathcal O_n$.
\end{lemma}

\begin{proof}
Let $\mathcal O_j=(Q_j,A,\cdot_j,\le_j),\ j=1,\dots,n$. Denote
$Q=Q_1\times\dots\times Q_n$ and 
$\mathcal O= (Q,A,\cdot,\le)$.

(i) Let $F_1, \dots , F_n$ be sets of final states used
for recognition of the languages $L_1, \dots , L_n$.
Put 
$F=F_1\times \dots \times F_n$
for the intersection $L_1\cap\dots\cap L_n$ and 
$$F=\{\, (q_1,\dots,q_n) \mid \text{ there exists } j\in\{1,\dots ,n\} \text{ such that }
q_j\in F_j\, \}$$
for the union $L_1\cup\dots\cup L_n$.
It is not hard to see that, in the both cases, $F$ is indeed 
an upward closed subset. 

(ii)
Let $L$ be a language recognized by  $(Q,A,\cdot,\le)$, i.e
let $F$ be an upward closed subset of  $Q$, and 
$i\in Q$ such that $L$ is recognized by 
$(Q,A,\cdot,\le,i,F)$.
Since $F= \bigcup_{p\in F} \cl{p}$, 
we see that $L=\bigcup_{p\in F} L_p$, where
$L_p$ is recognized by the ordered 
automata  $(Q,A,\cdot,\le,i,\cl{p})$.
Furthermore, for such $p$, we have $p=(p_1,\dots ,p_n)$ and we can write 
$\cl{p}=\cl{p_1}\times \dots \times \cl{p_n}$.  
Let $i=(i_1,\dots , i_n)$ and let $L_{(p,j)}$ be a language recognized
by the ordered automaton $(Q_j,A,\cdot_j,\le_j,i_j,\cl{p_j})$.
Then one can check that $L_p=L_{(p,1)}\cap \dots \cap L_{(p,n)}$. 
\qed\end{proof}

Also the following construction is useful. 

\begin{definition}\label{d:disjoint-union}
Let $I=\{1,\dots , n\}$  be a non\=/empty finite set and, for each $j\in I$,
let $\mathcal Q_j=(Q_j,A,\cdot_j,\le_j)$ be an ordered semiautomaton. 
We define the {\em disjoint union} $\mathcal Q=(Q,A,\cdot,\le)$
of ordered semiautomata
$\mathcal Q_1, \dots ,\mathcal Q_n$
in the following way:
\begin{itemize}
\item $Q=\{\, (q,j) \mid j\in I, q\in Q_j\, \}$, 
\item for each 
$a\in A$ and $(p,j), (q,k) \in Q$, we put
$ (q,j) \cdot a = (q \cdot_j a,j)$ and
\item we put
$(q,j) \le (p,k)$ if and only if $j=k$ and $q_j\le_j p_j$.
\end{itemize}
\end{definition}

Clearly, $L$ is recognized by a disjoint union of ordered 
semiautomata if and only if  it is recognized
by some of them.
A further useful notion is a homomorphism of
ordered semiautomata.

\begin{definition}\label{d:morphisms}
Let $(Q,A,\cdot,\le)$ and $(P,A,\circ,\preceq)$ be ordered semiautomata 
and $\varphi :Q\rightarrow P$
be a mapping. Then $\varphi$ is called a 
{\em homomorphism} of ordered semiautomata if it is isotone and
$\varphi(q\cdot a)=\varphi (q)\circ a$ for all $a\in A$, $q\in Q$.
If there exists a surjective homomorphism of ordered semiautomata 
from $(Q,A,\cdot,\le)$ to $(P,A,\circ,\preceq)$, then 
we say that $(P,A,\circ,\preceq)$ is a {\em homomorphic image} of 
$(Q,A,\cdot,\le)$.
We say that $\varphi$ is {\em backward order preserving} if,
for every $p,q\in Q$, the
inequality $\varphi(p)\preceq\varphi(q)$ implies $p\le q$.
If the  homomorphism $\varphi$ is surjective and backward 
order preserving, then we say that 
$(Q,A,\cdot,\le)$ is  {\em isomorphic} to $(P,A,\circ,\preceq)$.
\end{definition}

In what follows, we use often simply $(P,A,\cdot,\le)$ instead of
$(P,A,\circ,\preceq)$.
Note that every  
backward order preserving mapping is injective. 

In the setting of the previous definition,
one can prove by induction with respect to the length
of words that 
for an arbitrary homomorphism $\varphi$ of semiautomata 
that  the equality 
$\varphi(q\cdot u)=\varphi(q) \circ u$ holds for 
every state $q\in Q$ and every word $u\in A^*$.

\begin{lemma}\label{l:morphic-image-odfa}
Let  an ordered semiautomaton $(P,A,\cdot,\le)$ be a homomorphic 
image of an ordered semiautomaton
$(Q,A,\cdot,\le)$
and $L$ be recognized by $(P,A,\cdot,\le)$. Then  $L$ is also recognized by 
$(Q,A,\cdot,\le)$. 
\end{lemma}
\begin{proof}
If $L$ is recognized by an ordered automaton
$\mathcal P=(P,A,\cdot,\le,i,F)$, with $F$ being an upward closed subset,
and $\varphi$ is a surjective homomorphism of a semiautomaton 
$(Q,A,\cdot,\le)$ onto the semiautomaton $\ol{\mathcal P}=(P,A,\cdot,\le)$, 
then
we can choose some $i'\in Q$ such that $\varphi(i')=i$ 
and we can consider 
$F'=\{\, q\in Q \mid \varphi(q)\in F\,\}$.
Now $F'$ is an upward closed subset in $(Q,\le)$, because $\varphi$ is an isotone mapping and $F$ is upward closed.

Moreover, for an arbitrary $u\in A^*$, 
the following is a chain of equivalent
statements:
$$i'\cdot u \in F',\  \varphi(i'\cdot u)\in F, \ 
\varphi(i')\cdot u \in F,\  i\cdot u\in F,\ \text{and}\ u\in L\, .$$
The statement of the lemma follows.
\qed\end{proof}

\begin{definition}\label{d:trivial-odfa}
An ordered semiautomaton $(Q,A,\cdot,\le)$ is {\em trivial} if
$q\cdot a=q$ for all $q\in Q$ and $a\in A$, and $\le$ is the equality relation on $Q$.
In particular,
for a natural number $n$, 
we define the ordered semiautomaton 
$\mathcal{T}_n(A)=(I_n, A,\cdot,=)$, where $I_n=\{1,\dots ,n\}$, 
the transition function $\cdot $ is defined by the rule 
$j\cdot a=j$ for all $j\in I_n$ and $a\in A$.
\end{definition}
 It follows directly from the definition
that every trivial ordered semiautomaton is isomorphic to some 
$\mathcal{T}_n(A)=(I_n, A,\cdot,=)$.

\begin{lemma}\label{l:disjoint-via-product}
The disjoint union of ordered  
semiautomata $\mathcal Q_j=(Q_j,A,\cdot_j,\le_j)$, $j\in \{1,\dots ,n\}$, is 
  a homomorphic image of the product 
  $\mathcal Q_1 \times \dots \times \mathcal Q_n\times \mathcal{T}_n(A)$.
\end{lemma}

\begin{proof}
Clearly, the mapping defined by the rule
$$\varphi: (q_1,\dots,q_n,j) \mapsto (q_j,j),\ \text{for every}\  
q_1\in Q_1,\dots,q_n\in Q_n,\, j\in \{1,\dots ,n\}\, ,$$
is a surjective homomorphism of the considered semiautomata.
\qed\end{proof}

\begin{definition}\label{d:subautomaton}
Let $(Q,A,\cdot,\le)$ be an ordered semiautomaton and $P\subseteq Q$ be a non\=/empty subset.
If $p\cdot a \in P$ for every $p\in P$, $a\in A$, then
$(P,A,\cdot,\le)$  is called a {\em subsemiautomaton}
of $(Q,A,\cdot,\le)$.
\end{definition}

In the previous definition, the transition function and order on $P$ are restrictions 
of the corresponding data on the set $Q$ and so 
they are denoted by the same symbols.

Using the notions of a subsemiautomaton and a homomorphic image,
we can formulate the minimality of the canonical ordered 
semiautomaton in a bit precise way.

\begin{lemma}\label{l:minimality-canonical-dfa}
Let $(Q,A,\cdot,\le)$ be an ordered semiautomaton recognizing 
the language $L$.
Then the canonical ordered semiautomaton $\ol{\mathcal{D}_L}$ 
is a homomorphic image of some subsemiautomaton of 
$(Q,A,\cdot,\le)$.
\end{lemma}

\begin{proof}
Let $L$ be recognized by the ordered automaton 
$\mathcal A=(Q,A,\cdot,\le,i,F)$.
 It is easy  to see that the subset 
 $Q'=\{\,i\cdot u \mid u\in A^*\,\}$ 
 constructed in Lemma~\ref{l:canonical-dfa}
 forms a subsemiautomaton of $(Q,A,\cdot,\le)$. 
 Furthermore, we defined there
 the mapping $\varphi : Q'\rightarrow D_L$
 by the rule $\varphi(q) =u^{-1}L$, where $q=i\cdot u$.
 This mapping $\varphi$ is a surjective homomorphism of ordered
 semiautomata.
\qed\end{proof}

We say that an ordered semiautomaton $(Q,A,\cdot, \le)$ is 
{\it 1\=/generated} if
there exists a state $i\in Q$ such that 
$Q=\{\,i\cdot u \mid u\in A^*\,\}$.

\begin{lemma}\label{l:reconstruction-product}
Let $(Q,A,\cdot \le)$ be a 1\=/generated ordered semiautomaton. 
Then this semiautomaton is isomorphic to 
a subsemiautomaton of a product of the 
canonical ordered semiautomata of 
languages recognized by 
the ordered semiautomaton $(Q,A,\cdot \le)$. 
\end{lemma}

\begin{proof}
Let $\mathcal A=(Q,A,\cdot \le)$ be a 1\=/generated ordered 
semiautomaton, i.e.
$Q=\{\,i\cdot u \mid u\in A^*\,\}$ for some $i\in Q$.
For each $q\in Q$, we consider the ordered automaton 
$\mathcal{A}_q=(Q,A,\cdot,\le,i,\cl q)$.
This automaton recognizes the language 
$L_{\mathcal{A}_q}$, 
which we denote by $L(q)$. 
Using Lemma~\ref{l:minimality-canonical-dfa}, there is 
a surjective homomorphism $\varphi_q : \mathcal A \rightarrow 
\ol{\mathcal D_{L(q)}}$ of ordered semiautomata, because
$\mathcal A$ is 1\=/generated and thus $Q'=Q$ here.

Assume that $\mathcal A$ has $n$ states and denote them by $q_1,\dots ,q_n$.
We can consider the product of the canonical ordered semiautomata 
$\ol{\mathcal{D}_{L(q_k)}}=(D_{L(q_k)},A,\cdot_{q_k},\subseteq_{q_k})$
for all $k=1,\dots ,n$, 
i.e. $\ol{\mathcal{D}_{L(q_1)}} \times \dots \times \ol{\mathcal{D}_{L(q_n)}}$.
Moreover, since we have the mapping $\varphi_q$ for each $q\in Q$, 
we can define a mapping $\varphi : Q \rightarrow  
{D}_{L(q_1)}\times \dots \times {D}_{L(q_n)}$ by the rule
$\varphi(p)=(\varphi_{q_1}(p), \dots ,\varphi_{q_n}(p) )$, for $p\in Q$.
To prove the statement of the lemma it is enough to 
show that this mapping $\varphi$ is an backward order preserving 
homomorphism of ordered
semiautomata.

Let $p\le q $ hold in $Q$. For each $r\in Q$ the homomorphism 
$\varphi_r$ is isotone and hence $\varphi_r(p)\subseteq_r\varphi_r(q)$. 
Thus we get $\varphi(p)\le\varphi(q)$
and we see that $\varphi$ is an isotone mapping. 
In the similar way one can check that $\varphi(p\cdot a)=\varphi(p) \cdot a$ 
for every $p\in Q$ and $a\in A$. 
These facts mean that $\varphi$ is a homomorphism of ordered semiautomata.

Now assume that $p$ and $q$ are two states in $Q$ such that $\varphi(p)\le \varphi(q)$.
Then we have $\varphi_p(p) \subseteq_p \varphi_p(q)$ 
in $\ol{\mathcal{D}_{L(p)}}$.
By the definition of the mapping $\varphi_p$ given 
in Lemma~\ref{l:canonical-dfa},
for each $r\in Q$, we have $\varphi_p(r)= \nL_{\mathcal A_p, r}$. 
In particular, we can write 
$\nL_{\mathcal A_p, p}\subseteq \nL_{\mathcal A_p, q}$.
Since $p$ is a final state in $\mathcal A_p$, we have 
$\lambda\in \nL_{\mathcal A_p, p}$ 
and therefore $\lambda\in \nL_{\mathcal A_p, q}$, which means that 
$q$ is a final state in $\mathcal{A}_p$ too.
In other words $q\in \cl p$, i.e. $q\ge p$.
Thus the mapping $\varphi$ is backward order preserving.
\qed\end{proof}

\begin{lemma}\label{l:reconstruction-union}
Let $\mathcal A=(Q,A,\cdot \le)$ be an ordered semiautomaton. 
Then the semiautomaton $\mathcal A$
is a homomorphic image of a disjoint union of its 1\=/generated subsemiautomata. 
\end{lemma}

\begin{proof}
For every $q\in Q$, we consider the subset of $Q$ given by 
$Q_q=\{\, q\cdot u \mid u\in A^*\,\}$ 
consisting from all states reachable  
from $q$. Clearly, $Q_q$ form a 1\=/generated subsemiautomaton of  
$\mathcal A =(Q,A,\cdot, \le)$.
We consider disjoint union of all these semiautomata.
The set of all states of this ordered semiautomaton is 
$P=\{\,(p,q)\mid \text{there exists } u\in A^* \text{ such that } q\cdot u =p\,\}$.
We show that the mapping 
$\varphi :P \rightarrow Q$ given by the rule $\varphi((p,q))=p$
is a surjective homomorphism of ordered semiautomata.
Indeed, for $a\in A$, we have $(p,q)\cdot a= (p\cdot a,q)$, 
and hence
$\varphi((p,q)\cdot a)=\varphi((p\cdot a,q))=p\cdot a=\varphi((p,q))\cdot a$.
Moreover, $(p,q)\le (p',q')$ in  $P$ implies $q=q'$ and $p\le p'$, 
which means that $\varphi((p,q))\le\varphi((p,q'))$.
Finally, the surjectivity follows from
the fact $\{(q,q) \mid q\in Q\} \subseteq P$. 
 \qed\end{proof}

Since positive $\mathcal C$\=/varieties of languages are closed 
under taking preimages in morphisms
from the category $\mathcal C$
we need a construction on ordered
semiautomata which enables the 
recognition of 
such languages.

\begin{definition}\label{d:f-subautomaton}
Let $f: B^* \rightarrow A^*$ be a homomorphism and 
$\mathcal A=(Q,A,\cdot,\le)$ be an ordered semiautomaton.
By $\mathcal A^f$ we denote the semiautomaton
$(Q,B,\cdot^f,\le)$ where $q\cdot^f b= q\cdot f(b)$ 
for every  $q\in Q$ and $b\in B$.
We speak about $f${\em\=/renaming} of $\mathcal A$ and
we say that  $(P,B,\circ,\preceq)$ is an $f${\em\=/subsemiautomaton}
of $(Q,A,\cdot,\le)$ if it is a subsemiautomaton of 
$\mathcal A^f$.
In other words, if
$P\subseteq Q$, the order $\preceq$ is the restriction 
of $\le$, and $\circ$ is a restriction of $\cdot^f$.
\end{definition}

If we consider $f=\id_A : A^*\rightarrow A^*$, we can see 
that $(Q,A,\cdot,\le)^{\id}=(Q,A,\cdot,\le)$ and  that
$\id$\=/subsemiautomata of $(Q,A,\cdot,\le)$ 
are exactly its subsemiautomata.

\begin{lemma}\label{l:f-automaton-preimage}
Consider a homomorphism $ f:B^*\rightarrow A^*$ of monoids. 

(i) Let $L$ be a regular language which 
is recognized by an ordered automaton 
$(Q,A,\cdot,\le, i, F)$. 
Then the automaton $\mathcal B= (Q,B,\circ,\le, i, F)$, 
where $q\circ b= q\cdot f(b)$ for every 
$q\in Q$, $b\in B$, recognizes the language  
$f^{-1}(L)$.

(ii) Let $\mathcal A=(Q,A,\cdot,\le)$ be an ordered semiautomaton.
If $K\subseteq B^*$ is recognized by the ordered semiautomaton 
$\mathcal A^f$,
then there exists a language $L\subseteq A^*$ recognized by 
$\mathcal A$ such that
$K=f^{-1}(L)$.
\end{lemma}
\begin{proof}

(i)
For every $u\in B^*$, 
we have the following chain of equivalent formulas:
$$u\in f^{-1}(L),\ f(u)\in L,\
i\cdot f(u)\in F,\ i\circ u\in F, \
\text{and} \ u\in \nL_{\mathcal B}\, .$$

(ii)
If $K\subseteq B^*$ is recognized by the ordered semiautomaton 
$\mathcal A^f$ then there is a state $i\in Q$
 and an upward closed subset $F\subseteq Q$ such that $K=L_{\mathcal{B}}$, where
 $\mathcal B=(Q,B,\cdot^f,\le,i,F)$. Now we consider 
 $L=L_{\mathcal A'}$ , where
 $\mathcal A'=(Q,A,\cdot,\le,i,F)$. 
 Now the equality $K=f^{-1}(L)$ follows from 
 Lemma~\ref{l:f-automaton-preimage} (i).
 \qed\end{proof}
 
\begin{lemma}\label{l:f-subautomaton-vs-HSP}
Let $f$ be an arbitrary homomorphism $f:B^*\rightarrow A^*$.

(i) If an ordered semiautomaton 
  $\mathcal B=(P,A,\circ,\preceq)$ is a homomorphic image of 
  an ordered semiautomaton $\mathcal A= (Q,A,\cdot,\le)$, then 
  $\mathcal B^f$ is a homomorphic image of the ordered 
  semiautomaton $\mathcal A^f$.

  (ii) If an ordered semiautomaton  $\mathcal B=
  (P,A,\circ,\preceq)$ is a subsemiautomaton of an ordered 
  semiautomaton $\mathcal A=(Q,A,\cdot,\le)$ 
  then $\mathcal B^f$ is a subsemiautomaton of $\mathcal A^f$.

(iii) 
  If an ordered semiautomaton 
  $\mathcal B=\mathcal A_1 \times \dots \times \mathcal A_n$ is a product of a family of ordered semiautomata 
  $\mathcal A_1$, \dots , $\mathcal A_n$,
  then $\mathcal B^f=\mathcal A_1^f \times \dots \times \mathcal A_n^f$.  
\end{lemma}

\begin{proof}
(i) Let $\varphi$ be a surjective homomorphism from an  
 ordered semiautomaton $(Q,A,\cdot,\le)$
 onto $(P,A,\circ,\preceq)$. Then $\varphi$ is a isotone mapping from 
 $(Q,\le)$ to $(P,\preceq)$.
 The states and the  order in the semiautomaton
 $(Q,A,\cdot,\le) ^f$ resp.  $(P,A,\circ,\preceq)^f$
 are unchanged and hence $\varphi$ is an isotone mapping from  
 the ordered semiautomaton $(Q,A,\cdot,\le) ^f$ onto 
  $(P,A,\circ,\preceq)^f$.
 Now let $b\in B$ be an arbitrary letter and $q\in Q$ be an arbitrary state. 
 Then $\varphi(q \cdot^f b)=\varphi(q\cdot f(b))=\varphi(q) \circ f(b)=
 \varphi(q)\circ^f b$. Therefore $\varphi$ is a surjective 
 homomorphism of ordered semiautomata.
 
(ii) This is clear.
 
(iii) Let 
 $\mathcal A_j=(Q_j,A,\cdot_j,\le_j)$ be an ordered semiautomaton for every 
 $j=1,\dots , n$. 
 Let $(P,A,\circ,\preceq)$ be the product
 $\mathcal A_1\times \dots \times \mathcal A_n$
 and $(R,B,\diamond,\sqsubseteq)$ be the product of ordered 
 semiautomata $\mathcal A_1 ^f\times \dots \times \mathcal A_n^f$.
 Directly from the definitions we have that $P=R=Q_1\times \dots \times Q_n$ 
 and $\preceq\ =\ \sqsubseteq$.
 Furthermore, for an arbitrary element $(q_1,\dots ,q_n)$ from the set $P=R$, we have
$$
(q_1,\dots ,q_n)  \circ^f b = (q_1,\dots ,q_n)  \circ f(b) = $$ 
$$ = (q_1\cdot_1 f(b), \dots , q_n\cdot_n f(b)) 
= (q_1\cdot_1^f b, \dots ,q_n\cdot_n^f b )$$
 which is equal to 
 $(q_1,\dots ,q_n) \diamond b$. 
 This means that the action by each letter $b$ is defined in 
the ordered semiautomaton $(P,A,\circ,\preceq)^f$ 
in the same way as in the ordered semiautomaton 
$(R,B,\diamond,\sqsubseteq)$. 
\qed\end{proof}


\section{Eilenberg Type Correspondence for 
Positive $\mathcal C$\=/Varieties of Ordered Semiautomata}
\label{s:eilenberg}


\begin{definition}\label{d:variety-automata}
Let $\mathcal C$ be a category of homomorphisms.
A positive $\mathcal C$\=/variety of ordered semiautomata $\mathbb{V}$
associates to every non\=/empty finite alphabet $A$ a
class $\mathbb{V}(A)$ of ordered semiautomata over 
$A$ in such a way  that
\begin{itemize}
\item
$\mathbb{V}(A)\not = \emptyset$ is closed under disjoint unions
and direct products of non\=/empty finite families, and homomorphic images,
\item  $\mathbb {V}$ is closed under $f$\=/subsemiautomata for
all $(f: B^* \rightarrow A^*)\in \mathcal C$.
\end{itemize}
\end{definition}

\begin{remark}\label{r:trivial-variety}
We define $\mathbb{T}(A)$ as a class of all trivial ordered
semiautomata over an alphabet $A$, i.e.
$\mathbb{T}(A)$ contains all semiautomata $\mathcal T_n(A)$ and all their 
isomorphic copies.
By Lemma~\ref{l:disjoint-via-product},
the first condition in the definition of positive $\mathcal C$\=/variety 
of ordered semiautomata
can be written equivalently in the following way: 
$\mathbb T (A) \subseteq \mathbb{V}(A)$ and 
$\mathbb{V}(A)$ is closed under direct products of non\=/empty finite
families and homomorphic images.
In particular, the class of all
trivial ordered semiautomata $\mathbb T$ forms the smallest positive 
$\mathcal C$\=/variety of ordered semiautomata
whenever the considered category $\mathcal C$ contains all isomorphisms.

As mentioned in the introduction, Theorem~\ref{t:eilenberg-ordered}
has already been proved in special cases.
The technical difference is that \'Esik and Ito in~\cite{esik-ito} used disjoint union 
of automata and Chaubard, Pin and Straubing~\cite{pin-str} 
did not use this construction because they used trivial automata instead of them.
\end{remark}

Now we are ready to state the 
Eilenberg type correspondence for positive $\mathcal C$\=/varieties of ordered
semiautomata.

For each positive $\mathcal C$\=/variety of ordered semiautomata $\mathbb V$, we denote by
$\alpha(\mathbb V)$ the class of regular languages given by the following formula
$$(\alpha(\mathbb V)) (A)=
\{\,L\subseteq A^* \mid \ol{\mathcal{D}_L}\in \mathbb V (A)\, \}\, .$$

For each positive $\mathcal C$\=/variety of regular languages $\mathcal L$,
we denote by $\beta(\mathcal L)$ the positive $\mathcal C$\=/variety of 
ordered semiautomata 
generated by all ordered semiautomata $\ol{\mathcal{D}_L}$, 
where $L\in \mathcal L (A)$ for some 
alphabet $A$.

The following result can be obtained as a combination of Theorem 5.1 of \cite{pin-straubing}
 and the main result of \cite{pin-str}. We show a direct and detailed proof here.

\begin{theorem}\label{t:eilenberg-ordered}
Let $\mathcal C$ be a category of homomorphisms.
The mappings $\alpha$ and $\beta$ are mutually inverse isomorphisms
between the lattice of all positive $\mathcal C$\=/varieties of 
ordered semiautomata and the lattice
of all positive $\mathcal C$\=/varieties of regular languages.
\end{theorem}
\begin{proof}
 First of all, we fix a category of homomorphism $\mathcal C$ for the whole proof.
The proof will be done when we show the following statements:
\begin{enumerate}
\item $\alpha$ is correctly defined,
i.e. for every positive $\mathcal C$\=/variety of ordered semiautomata
$\mathbb V$, the class 
$\alpha(\mathbb V)$ is a positive $\mathcal C$\=/variety of languages.
\item $\beta$  is correctly defined,  
i.e. for every positive $\mathcal C$\=/variety
of regular languages $\mathcal L$, the class
$\beta(\mathcal L)$ is a positive $\mathcal C$\=/variety of ordered semiautomata.
\item $\beta \circ \alpha = \mathrm{id}$, i.e.
for each positive $\mathcal C$\=/variety of ordered semiautomata $\mathbb V$ 
we have
$\beta (\alpha(\mathbb V)) = \mathbb V$.
\item $\alpha\circ \beta =\mathrm{id}$,  i.e. 
for each positive $\mathcal C$\=/variety of languages $\mathcal L$, we have
$\alpha (\beta(\mathcal L)) = \mathcal L$.
\end{enumerate}

We prove these facts in separate lemmas. 
The exception is the second item which trivially follows from
the definition of the mapping $\beta$. 
Before the formulation of these lemmas we prove some technicalities.

\begin{lemma}\label{l:formula-for-alpha}
For each positive $\mathcal C$\=/variety of ordered semiautomata 
$\mathbb V$ and an alphabet $A$ we have 
$$(\alpha(\mathbb V)) (A)=\{\,L\subseteq A^* \mid 
\exists\, \mathcal A=(Q,A,\cdot,\le, i, F) :  
L=\nL_{\mathcal A}\ \text{ and }\  \ol{\mathcal A} \in \mathbb V(A)\, \}\, .$$
\end{lemma}
\begin{proof}
The inclusion ``$\subseteq$'' is trivial, because one can take for the ordered 
automaton $\mathcal A$ the canonical automaton $\mathcal{D}_L$.
To prove the opposite inclusion, let $L=\nL_{\mathcal A}$, where
$\mathcal A=(Q,A,\cdot,\le, i, F)$ with $\ol{\mathcal{A}} \in \mathbb V(A)$.
By Lemma~\ref{l:minimality-canonical-dfa} and the assumption that 
$\mathbb V$ is closed under taking subsemiautomata and homomorphic images,
we have that $\ol{\mathcal{D}_L}\in \mathbb V(A)$.
Therefore $L\in (\alpha(\mathbb V)) (A)$.
\qed\end{proof}

\begin{lemma} If $\mathbb V$ is a positive $\mathcal C$\=/variety of ordered semiautomata, then
$\alpha(\mathbb V)$ is a positive $\mathcal C$\=/variety of regular languages.
\end{lemma}
\begin{proof}

We need to prove that $(\alpha(\mathbb{V}))(A)$ is closed under 
taking intersections, unions and quotients.
Secondly, we must show the closure property with respect 
to taking preimages in morphisms from the category
$\mathcal C$.

For each $A$, the class $(\alpha(\mathbb V)) (A)$
given by the formula from Lemma~\ref{l:formula-for-alpha}
is closed under unions and intersections of finite families,
since $\mathbb V(A)$ is closed under  products  of finite families
(see Lemma~\ref{l:product-arbitrary-subset}).
The class $(\alpha(\mathbb V)) (A)$ is also closed under 
quotients, since we can change initial and final states freely
(see Lemma~\ref{l:quotients}).

Furthermore, by  Lemma~\ref{l:f-automaton-preimage},
we see the following observation.
Since $\mathbb V$ is closed under taking 
$f$\=/subsemiautomata for each homomorphism  $f$ from $\mathcal C$, the class  
$\alpha(\mathbb V)$ is closed under preimages in the same homomorphisms.
\qed\end{proof}

All three constructions -- direct product, homomorphic image and 
subsemiautomaton --
are standard constructions of universal algebra.
From the general theory (see e.g. Burris and 
Sankappanavar~\cite{burris1981course}) 
we want to use only the fact that 
if one needs to generate the smallest class closed with respect to all three constructions together 
and containing a class
$\XX$, then it is enough to 
consider a homomorphic images of subalgebras in products of algebras form $\XX$.
Note that from this point of view, an alphabet $A$ is fixed, and $A$ serves as a set of unary function symbols.
Then a semiautomaton over $A$ is a unary algebra.

For a class of ordered semiautomata $\XX$ over a fixed alphabet $A$ we denote
by 
\begin{itemize}
\item $\OH\XX$  the class of all homomorphic images of ordered semiautomata from $\XX$,
\item $\OI\XX$  the class of all isomorphic copies of ordered semiautomata from $\XX$,
\item $\OS\XX$ the class of all subsemiautomata of ordered semiautomata from $\XX$,
\item $\OP\XX$ the class of all products of non\=/empty finite families of
ordered semiautomata from $\XX$,
\item $\OD\XX$ the class of all disjoint unions of non\=/empty finite families of
ordered semiautomata from $\XX$.
\end{itemize}
It is clear that the operators $\OH$, $\OI$ and $\OS$
are idempotent, i.e. for each class of ordered semiautomata 
$\XX$ we have $\OH\OH\XX=\OH\XX$ etc.
Furthermore,  $\OI\OP\OP=\OI\OP$ and $\OI\OD\OD=\OI\OD$.

\begin{lemma}\label{l:operators-properties}
 For each class $\XX$ of ordered 
 semiautomata over a fixed alphabet $A$, we have:
$$\OD\XX \subseteq \OH\OP(\XX\cup \mathbb{T}(A))$$
and 
$$\OP\OH\XX \subseteq \OH\OP\XX,\quad 
\OS\OH\XX \subseteq \OH\OS\XX, \quad \OP\OS
\XX \subseteq \OS\OP\XX\, .$$
\end{lemma}
\begin{proof} The first property follows from Lemma~\ref{l:disjoint-via-product}.
 The other properties are well\=/known facts from universal algebra, 
 see e.g.~\cite[Chapter II, Section 9]{burris1981course} --
 a modification for the ordered case is straightforward. 
\qed\end{proof}

\begin{lemma}\label{l:hsp}
For each positive $\mathcal C$\=/variety of regular languages 
$\mathcal L$ we have 
$$(\beta(\mathcal L))(A)=
\OH\OS\OP(\,\{\ol{\,\mathcal{D}_L}\mid L\in \mathcal{L}(A)\,\}
\cup\mathbb{T}(A)\,)\, .$$
\end{lemma}
\begin{proof}
 For every alphabet $A$, we denote $\XX= \{\,\ol{\mathcal{D}_L} 
\mid L\in \mathcal{L}(A)\,\}\cup\mathbb{T}(A)$ and  we denote
by $\mathbb{L}(A)$ the 
right\=/hand side of the formula in the statement, i.e. 
$\mathbb{L}(A)=\OH\OS\OP\XX$.

Since $\beta(\mathcal L)$ is a positive $\mathcal C$\=/variety of ordered semiautomata,
we have $\mathbb{T}(A) \subseteq (\beta(\mathcal L))(A)$ by Remark~\ref{r:trivial-variety}. 
Therefore $\XX \subseteq (\beta(\mathcal L))(A)$ and the inclusion 
$\mathbb{L}(A)\subseteq 
(\beta(\mathcal L))(A)$ follows from the fact that $\beta(\mathcal L)(A)$ is closed under
operators  $\OH$, $\OS$ and $\OP$. 

To prove the opposite inclusion 
$(\beta(\mathcal L))(A) \subseteq \mathbb{L}(A)$, we first  prove that $\mathbb L$ is a positive
$\mathcal C$\=/variety of ordered semiautomata. 

By the first property in 
Lemma~\ref{l:operators-properties}, we get 
$$\OD(\mathbb{L}(A))=\OD\OH\OS\OP\XX\subseteq 
\OH\OP(\OH\OS\OP\XX)\, .$$ 
By the other properties of Lemma~\ref{l:operators-properties} and idempotency of the operators
we get $\OH\OP\OH\OS\OP\XX\subseteq 
\OH\OS\OP\XX=\mathbb{L}(A)$. Thus  
$\OD(\mathbb{L}(A)) \subseteq \mathbb{L}(A)$.  
In the same way one can prove another inclusions 
$\OH(\mathbb{L}(A)) \subseteq \mathbb{L}(A)$,
$\OS(\mathbb{L}(A)) \subseteq \mathbb{L}(A)$, 
$\OP(\mathbb{L}(A)) \subseteq \mathbb{L}(A)$.
Therefore $\mathbb{L}(A)$ is closed under all four operators 
$\OH$, $\OS$, $\OP$ and $\OD$.

It remains to prove that $\mathbb{L}$ is closed under $f$\=/renaming.
So, let $f: B^* \rightarrow A^*$ belong to $\mathcal C$. 
We need to show that $(Q,A,\cdot,\le)^f$
belongs to $\mathbb{L}(B)$ whenever $(Q,A,\cdot,\le)$ is from $\mathbb{L}(A)$. 
For an arbitrary set $\YY$ of ordered semiautomata 
over the alphabet $A$, 
we denote 
${\YY}^f=\{\,(Q,A,\cdot,\le)^f 
\mid (Q,A,\cdot,\le)\in \YY\,\}$.
Using this notation,  we need to show that 
$(\mathbb{L}(A))^f\subseteq \mathbb{L}(B)$.

At first, we show a weaker inclusion ${\XX}^f\subseteq \mathbb L(B)$.
Trivially $(\mathbb{T}(A))^f=(\mathbb{T}(B))$.
Now let $L$ be an arbitrary language from $\mathcal {L}(A)$. 
We consider 
$(D_L,A,\cdot,\subseteq)^f$ which is an ordered semiautomaton over $B$.
By Lemma~\ref{l:f-automaton-preimage},
the set $Z$ of all regular languages which are recognized 
by the ordered semiautomaton 
$(D_L,A,\cdot,\subseteq)^f$ contains only languages 
of the form $f^{-1}(K)$, where $K$ is recognized by 
$(D_L,A,\cdot,\subseteq)$. 
Since $\mathcal L(A)$ is closed under unions, intersections 
and quotients, every such language $K$ belongs 
to $\mathcal L(A)$
by Lemmas~\ref{l:quotients} and~\ref{l:quotients-in-canonical}.
This means that the set $Z$ 
is a subset of $\mathcal L (B)$ because $\mathcal L$ is closed under 
preimages in the homomorphism $f$.
Therefore $\mathbb L(B)$ contains all canonical ordered semiautomata of languages from $Z$.
Finally, $(D_L,A,,\cdot,\subseteq)^f$ can be reconstructed from these canonical 
ordered semiautomata of languages from $Z$: 
by Lemma~\ref{l:reconstruction-union}, 
the ordered semiautomaton $(D_L,A,\cdot,\subseteq)^f$ is a 
homomorphic image of a disjoint union 
of certain subsemiautomata which are isomorphic, 
by Lemma~\ref{l:reconstruction-product},
to subsemiautomata of products
of the canonical ordered semiautomata of languages from $Z$.
Hence $(D_L,A,\cdot,\subseteq)^f$ belongs to 
$\mathbb L(B)$ which is closed under homomorphic images, 
subsemiautomata, products and disjoint unions
as we proved above.
So, we proved ${\XX}^f\subseteq \mathbb L(B)$.

Now Lemma~\ref{l:f-subautomaton-vs-HSP} has the following 
consequences
$(\OH\YY)^f\subseteq \OH({\YY}^f)$, 
$(\OS\YY)^f\subseteq \OS({\YY}^f)$ and
$(\OP\YY)^f\subseteq \OP({\YY}^f)$ for an 
arbitrary set of ordered semiautomata $\YY$ over the 
alphabet $A$.
If we use all these properties we get 
$(\mathbb L(A))^f=(\OH\OS\OP\XX)^f\subseteq 
\OH\OS\OP({\XX}^f)\subseteq 
\OH\OS\OP (\mathbb L(B))=\mathbb L(B)$
because $\mathbb{L}(B)$ is closed under all three operators.
Hence
$\mathbb L$ is closed under taking $f$\=/subsemiautomata and 
therefore
$\mathbb L$ is a positive $\mathcal C$\=/variety of ordered semiautomata.

The inclusion $(\beta(\mathcal L))(A)\subseteq  \mathbb{L}(A)$
follows from the definition of $\beta(\mathcal L)$ which is 
the smallest positive
$\mathcal C$\=/variety of ordered semiautomata containing 
$\XX$.
Since the opposite inclusion is also proved we have finish 
the proof of the lemma. 
\qed\end{proof}

\begin{lemma}\label{l:eilenberg-part3}
For each positive $\mathcal C$\=/variety of ordered semiautomata 
$\mathbb V$ we have
$\beta (\alpha(\mathbb V)) = \mathbb V$.
\end{lemma}
\begin{proof}
Let $A$ be an arbitrary alphabet.
By Lemma~\ref{l:hsp},
$$(\beta (\alpha(\mathbb V))) (A)=\OH\OS\OP 
(\,\{\ol{\mathcal{D}_L}\mid L\in (\alpha(\mathbb V))(A)\}
\cup\mathbb{T}(A)\,)\, .$$
We denote $\XX=
\{\,\ol{\mathcal{D}_L}\mid L\in (\alpha(\mathbb V))(A)\,\}$.
If we use the definition of the mapping $\alpha$ then we see 
that
$\XX=\{\,(Q,A,\cdot,\le) \in \mathbb{V}(A)\mid \exists 
L\subseteq A^* : 
\ol{\mathcal{D}_L}=(Q,A,\cdot,\le)\,\}$.
In particular $\XX\subseteq \mathbb V (A)$.
Since we also have $\mathbb T(A)\subseteq \mathbb V(A)$
we see that $\XX \cup\mathbb{T}(A)\subseteq 
\mathbb V(A)$.
Hence  
$$(\beta (\alpha(\mathbb V))) (A)=\OH\OS\OP(\, \XX 
\cup\mathbb{T}(A)\,)\subseteq \mathbb V(A)$$
because $\mathbb V(A)$ is closed under taking homomorphic images, 
subsemiautomata and products.

In the proof of Lemma~\ref{l:hsp} we already saw that 
every ordered semiautomaton $(Q,A,\cdot,\le)$ can be reconstructed 
from the canonical ordered automata 
of languages which are recognized by $(Q,A,\cdot,\le)$ by Lemma~\ref{l:reconstruction-product}
and~\ref{l:reconstruction-union}.
Therefore $\mathbb V(A) \subseteq 
\OH\OS\OP ( \XX \cup\mathbb{T}(A))$ and we proved the 
equality
$\mathbb V(A) = \OH\OS\OP ( \XX \cup\mathbb{T}(A))$,
which means that $\beta (\alpha(\mathbb V)) = \mathbb V$.
\qed\end{proof}

\begin{lemma}\label{l:eilenberg-part4}
Let  $\mathcal L$ be a positive $\mathcal C$\=/variety of languages $\mathcal L$.
Then $\alpha (\beta(\mathcal L)) = \mathcal L$.
\end{lemma}
\begin{proof}
We want to prove that for every $A$ the equality 
$(\alpha (\beta(\mathcal L))) (A) = \mathcal L (A)$ holds.
Let $L\in \mathcal L(A)$ be an arbitrary language.
By the definition of the mapping $\beta$, we have 
$\ol{\mathcal{D}_L}\in (\beta(\mathcal L)) (A)$.
Therefore, by definition of $\alpha$, we have 
$L\in \alpha (\beta(\mathcal L)) (A)$
and we have proved the inclusion ``$\supseteq$''.

To prove the opposite one, 
let $K\in \alpha ((\beta(\mathcal L))) (A)$ be an 
arbitrary language.
Then there is an ordered automaton 
$\mathcal A=(Q,A,\cdot ,\le, i, F)$
such that $K=\nL_{\mathcal A}$ 
and 
$\ol{\mathcal A}\in (\beta(\mathcal L))(A)
=\OH\OS\OP(\,\{\,\ol{\mathcal{D}_L}\mid L
\in \mathcal{L}(A)\,\}\cup\mathbb{T}(A)\,)$.
If $K$ is recognized by $\mathcal A$, where 
$\ol{\mathcal A}\in \OH\XX$ for the
class of ordered semiautomata 
$\XX=\OS\OP(\,\{\,\ol{\mathcal{D}_L}\mid L\in 
\mathcal{L}(A)\,\}\cup\mathbb{T}(A)\,)$, then
there is an ordered automaton $\mathcal B$ 
such that $\ol{\mathcal A}$ is a homomorphic image of 
$\ol{\mathcal B}\in\XX$.
By Lemma~\ref{l:morphic-image-odfa},
the language $K$ is recognized 
by $\ol{\mathcal B}$.
Thus we can assume that $\mathcal A$ belongs to 
$\OS\OP(\,\{\,\ol{\mathcal{D}_L}\mid L\in 
\mathcal{L}(A)\,\}\cup\mathbb{T}(A)\,)$.
In the same way we can also assume that 
$\ol{\mathcal A}$ belongs to 
$\OP(\,\{\,\ol{\mathcal{D}_L}\mid L\in \mathcal{L}(A)\,\}
\cup\mathbb{T}(A)\,)$.
By Lemma~\ref{l:product-arbitrary-subset}, we know that 
$K$ is a finite union of 
finite intersections of languages which are recognized by 
ordered semiautomata
from the class $\{\ol{\,\mathcal{D}_L}\mid 
L\in \mathcal{L}(A)\,\}\cup\mathbb{T}(A)$.
Furthermore, trivial semiautomata recognize only languages $\emptyset$ 
and $A^*$ which belong to every  $\mathcal L(A)$, hence 
we may consider $\{\,\ol{\mathcal{D}_L}\mid 
L\in \mathcal{L}(A)\,\}$ instead of $\{\,\ol{\mathcal{D}_L}\mid 
L\in \mathcal{L}(A)\,\}\cup\mathbb{T}(A)$ in the previous sentence. 
Since 
the canonical automaton $\mathcal{D}_L$
recognizes only finite unions of finite intersections of
quotients of the language $L$ 
(by Lemmas~\ref{l:quotients} 
and~\ref{l:quotients-in-canonical}), and since 
$\mathcal L(A)$ is closed under taking quotients, 
unions and intersections,
we see that $K$ belongs to $\mathcal L(A)$.
\qed\end{proof}

The previous lemma finishes the proof of 
Theorem~\ref{t:eilenberg-ordered}.
\qed\end{proof}


\section{$\mathcal C$\=/Varieties of Semiautomata}

For an ordered semiautomaton $(Q,A,\cdot,\le)$ we define its 
{\it dual}
$(Q,A,\cdot,\le)^{\Od}=(Q,A,\cdot, \le^{\Od})$ where 
$\le^{\Od}$ is the dual order to $\le$,
i.e. $p\le^{\Od} q$ if and only if $q\le p$.
Instead of the symbol $\le^{\Od}$ we usually use the symbol $\ge$.
Trivially, the resulting structure $(Q,A,\cdot,\ge)$ is also 
an ordered 
semiautomaton.
For a positive $\mathcal C$\=/variety of ordered semiautomata $\mathbb V$ 
we denote by $\mathbb V^{\Od}$ its dual, 
i.e for every alphabet $A$ we consider 
$\mathbb V^{\Od}(A)=\{\,(Q,A,\cdot,\le)^{\Od}\mid 
(Q,A,\cdot,\le)\in\mathbb V(A)\, \}$.
It is clear that $(\mathbb V^\Od)^{\Od}=\mathbb V$ and that 
$\mathbb V^{\Od}$ 
is a  positive $\mathcal C$\=/variety of ordered semiautomata.

We say that $\mathbb V$ is {\it selfdual} if 
$\mathbb V^{\Od}=\mathbb V$.
In other words, $\mathbb V$ is selfdual if and only if 
every $\mathbb V(A)$ is closed under taking duals of its members.
An alternative characterization follows.

\begin{lemma}\label{l:selfdual}
Let $\mathbb V$ be a positive $\mathcal C$\=/variety of ordered semiautomata.
Then $\mathbb V$ is selfdual if and only if for each alphabet $A$,
we have that:
$$(Q,A,\cdot,\le)\in \mathbb V (A) \text{ implies } 
(Q,A,\cdot,=)\in\mathbb V(A)\, .$$
\end{lemma}
\begin{proof}
If $\mathbb V^d=\mathbb V$ and
$(Q,A,\cdot,\le)\in\mathbb V(A)$ then we also have $(Q,A,\cdot,\ge)\in \mathbb V(A)$.
Now the ordered semiautomaton $(Q,A,\cdot,=)$ is 
isomorphic to a subsemiautomaton 
of the product of the ordered semiautomata
$(Q,A,\cdot,\le)$ and $(Q,A,\cdot,\ge)$, namely the 
subsemiautomaton with the set of states
$\{\,(q,q)\mid q\in Q\,\}$.

To prove the converse, it is enough to see that an arbitrary ordered 
semiautomaton 
$(Q,A,\cdot,\le)$ is a homomorphic image of the semiautomaton 
$(Q,A,\cdot,=)$: the identity mapping is 
a homomorphism of the considered order semiautomata.
\qed\end{proof}

Recall that a {\em $\mathcal C$\=/variety of regular languages} is 
a positive $\mathcal C$\=/variety of languages which is closed under 
taking complements.
The canonical ordered semiautomaton 
of the complement of a regular language $L$ is the
dual of the canonical ordered semiautomaton of $L$, 
i.e
$\overline{\mathcal D_{L^c}}= (\overline{\mathcal D_{L}})^\Od$. 
This easy observation helps to prove the following statement.

\begin{proposition}\label{p:variety-as-special-positive}
There is one to one correspondence between 
$\mathcal C$\=/varieties of regular languages and selfdual positive
$\mathcal C$\=/varieties of ordered semiautomata. 
\end{proposition}

\begin{proof}
 The mentioned correspondence is given by the pairs of the mappings 
 $\alpha$ and $\beta$ from Theorem~\ref{t:eilenberg-ordered}.
 For a selfdual positive
$\mathcal C$\=/variety of ordered semiautomata $\mathbb V$, we know that 
$(\alpha(\mathbb V))(A)$ is closed under 
complements. This means that $\alpha(\mathbb V)$ is a $\mathcal C$\=/variety 
of regular languages.
Therefore, it remains to show that, for an arbitrary  $\mathcal C$\=/variety 
of regular languages $\mathcal L$,
 the positive $\mathcal C$\=/variety of ordered semiautomata 
 $\beta(\mathcal L)$ is selfdual.
 By Lemma~\ref{l:hsp}, $(\beta(\mathcal L))(A)=
\OH\OS\OP(\,\{\,\ol{\mathcal{D}_L}\mid L\in \mathcal{L}(A)\,\}
\cup\mathbb{T}(A)\,)$, where the set 
$\{\,\ol{\mathcal{D}_L}\mid L\in \mathcal{L}(A)\,\}
\cup\mathbb{T}(A)$ is selfdual. 
However, for every selfdual class of semiautomata $\XX$, the classes 
of ordered semiautomata
$\OP\XX$, $\OS\XX$ and $\OH\XX$ are selfdual again.
Hence $\beta(\mathcal L)$ is selfdual.
\qed\end{proof}

Since every ordered semiautomaton $(Q,A,\cdot,\le)$ is a homomorphic image
of the ordered semiautomaton $(Q,A,\cdot,=)$ we can consider the notion of 
$\mathcal C$\=/varieties of semiautomata
instead of selfdual positive $\mathcal C$\=/varieties of ordered 
semiautomata: 
$\mathcal C$\=/varieties of semiautomata are classes
of semiautomata which are 
closed under taking  $f$\=/subsemiautomata, homomorphic images, disjoint unions 
and finite products.

Let $\mathbb A(A)$ be the class of all ordered semiautomata
over the alphabeth $A$. Notice that $\mathbb A$ forms 
the greatest positive $\mathcal C$\=/variety of ordered semiautomata
for each category $\mathcal C$.

If we have $\mathcal C$\=/variety of semiautomata $\mathbb V$
then we can consider all possible compatible orderings on these 
semiautomata and 
define the positive $\mathcal C$\=/variety of ordered semiautomata
$\mathbb V^{\Oo}$ in the following sense
$$\mathbb V^{\Oo}(A)=\{\, (Q,A,\cdot,\le)\in \mathbb A(A)
\mid (Q,A,\cdot)\in\mathbb V(A)\, \}\, .$$
Clearly, $\mathbb V^{\Oo}$ is selfdual.
Conversely, for a selfdual positive $\mathcal C$\=/variety of ordered 
semiautomata 
$\mathbb V$, we can consider 
$$\mathbb V^{\Ou} (A)=\{\, (Q,A,\cdot) \mid 
\text{ there is an order } 
\le 
\text{ such that } (Q,A,\cdot,\le) \in \mathbb V (A)\,\}\, .$$
Now two mappings $\mathbb V \mapsto \mathbb V^{\Oo}
$ and 
$\mathbb V \mapsto \mathbb V^{\Ou}$ are mutually inverse mappings 
between
$\mathcal C$\=/varieties of semiautomata and selfdual positive
$\mathcal C$\=/varieties of ordered semiautomata.

Using this easy correspondence, we obtain the following result
as the consequence of Proposition~\ref{p:variety-as-special-positive}.

\begin{theorem}\label{t:variety-version}
There is one to one correspondence between 
$\mathcal C$\=/varieties of regular languages and 
$\mathcal C$\=/varieties of semiautomata. 
\end{theorem} 

Note that this results can be also obtained by composing the results
by Pin, Straubing~\cite{pin-straubing} with those of 
Chaubard, Pin and Straubing \cite{pin-str}


\section{Examples}

In this section we present several instances of Eilenberg type correspondence. 
Some of them are just reformulations of examples already mentioned in existing literature.
In particular, the first three subsections correspond to pseudovarieties
of aperiodic, $\mathcal R$\=/trivial and $\mathcal J$\=/trivial monoids, 
respectively.
Also Subsection~\ref{s:ord-increasing} has a natural 
counterpart in pseudovarieties of ordered monoids satisfying the inequality $1\le x$.  
In all these cases, $\mathcal C$ is the category of all homomorphisms denoted by $\mathcal C_{all}$.
Nevertheless, we believe that these correspondences viewed from the perspective of varieties of (ordered)
semiautomata are of some interest. Another four subsections works with different categories $\mathcal C$ and 
Subsections~\ref{subs:synchro} and~\ref{subs:inserting} 
bring new examples of (positive) $\mathcal C$\=/varieties of (ordered) automata.

\subsection{Counter\=/Free Automata}

The star free languages were characterized by 
Sch{\"u}tzenberger~\cite{schutz} as languages having aperiodic 
syntactic monoids. 
Here we recall the subsequent characterization 
of McNaughton and Papert~\cite{McNaughton} by counter\=/free automata.

\begin{definition}
We say that a semiautomaton
$(Q,A,\cdot)$ is {\em counter\=/free} if,
for each $u\in A^*$, $q\in Q$ and $n\in\mathbb N$ such that $q\cdot u^n =q$, we have 
$q\cdot u = q$.
\end{definition} 

\begin{proposition}\label{p:counter-free}
The class of all counter\=/free semiautomata forms a variety of semiautomata.
\end{proposition}
\begin{proof}
It is easy to see that 
disjoint unions, subsemiautomata, products and $f$\=/renamings
of a counter\=/free semiautomata are again counter\=/free.

Let $\varphi : (Q,A,\cdot) \rightarrow (P,A,\circ)$ 
be a surjective homomorphism of semiautomata and let $(Q,A,\cdot)$ be counter\=/free.
We prove that also $(P,A,\circ)$ is a counter\=/free semiautomaton.

Take $p\in P, u\in A^*$ and $n\in\mathbb N$ 
such that  $p\circ u^n= p$.
Let $q\in Q$ be an arbitrary state such that $\varphi (q)=p$.
Then, for each $j\in\mathbb N$, we have $\varphi(q\cdot u^{jn}) 
=p\circ u^{jn} = p$.
Since the set
$\{q, q\cdot u^n, q\cdot u^{2n}, \dots \}$ is finite,
there exist $k,\ell \in\mathbb N$ such that $q\cdot u^{kn} = q \cdot u^{(k+\ell)n}$.
If we take $r=q \cdot u^{kn}$, then $r\cdot u^{\ell n}=r$.
Since $(Q,A,\cdot)$ is counter\=/free, we get
$r\cdot u=r$. 
Consequently, $p\circ u=\varphi(r)\circ u=\varphi(r\cdot u)=\varphi(r)=p$.
\qed\end{proof}

The promised link between languages and automata follows.

\begin{proposition}[McNaughton, Papert~\cite{McNaughton}]
 Star free languages are exactly the languages
 recognized by counter\=/free semiautomata.
\end{proposition}

Note that this characterization is effective, although  
testing whether a regular language given by a DFA is aperiodic is
even PSPACE\=/complete problem by Cho and Huynh~\cite{cho}.

\subsection{Acyclic Automata}

The {\it content} $\mathrm{c}(u)$ of a word $u\in A^*$ is the 
set of all letters occurring in $u$. 

\begin{definition}
We say that a semiautomaton $(Q,A,\cdot)$ is {\em acyclic} if,
for every $u\in A^+$  and $q\in Q$ such that $q\cdot u =q$,
we have $q\cdot a = q$ for every $a\in \mathrm{c}(u)$.
\end{definition}

Note that one of the conditions in Simon's characterization of piecewise testable 
languages is that 
a minimal DFA is acyclic -- see \cite{simon-pw}.

One can prove the following proposition in a similar way as in the case of counter\=/free 
semiautomata.

\begin{proposition}
The class of all acyclic semiautomata forms a variety of semiautomata.
\end{proposition}

According to Pin~\cite[Chapter 4, Section 3]{pin1986varieties}, 
a semiautomaton  $(Q,A,\cdot)$ is called {\em extensive} if there  
exists a linear order $\, \preceq\, $ on 
$Q$ such that
$(\,\forall\, q\in Q,\, a\in A\,) \ q\preceq q\cdot a$.
Note that such an order need not to be compatible with actions of letters.
One can easily show that a semiautomaton is acyclic if and only if it is extensive.
We prefer to use the term acyclic, since we consider extensive
actions by letters (compatible with ordering of a semiautomaton) later in the paper.
Anyway, testing whether a given semiautomaton is acyclic can be decided
using the breadth\=/first search algorithm.

\begin{proposition}[Pin \cite{pin1986varieties}]
 The languages over the alphabet $A$ accepted by acyclic semiautomata
 are exactly 
 disjoint unions
of the languages of the form
$$A_0^*a_1A_1^*a_2 A_2^* \dots A_{n-1}^*a_n A_n^*\ 
\text{ where  }\ a_i\not\in A_{i-1}\subseteq A\,\text{ for }\ i=1,\dots,n\,.$$
\end{proposition}

Note that the languages above are exactly those
having $\mathcal R$\=/trivial syntactic monoids

\subsection{Acyclic Confluent Automata}

In our paper~\cite{dlt13-klima} concerning piecewise testable languages, 
we introduced a certain condition on automata 
being motivated by the terminology 
from the theory of rewriting systems. 

\begin{definition}
We say that a semiautomaton $(Q,A,\cdot)$ is {\em confluent}, 
if for each state $q\in Q$ and every
pair of words $u,v\in A^*$,
there is a word $w\in A^*$ such that 
$\mathrm{c}(w)\subseteq \mathrm{c}(uv)$ and
$(q\cdot u)\cdot 
w=(q\cdot v)\cdot w$.
\end{definition}

In~\cite{dlt13-klima}, this definition was studied in the context of 
acyclic (semi)automata, in which case several equivalent conditions were described.
One of them can be rephrased in the following way.
\begin{lemma}\label{l:acyclic-confluent}
 Let $(Q,A,\cdot)$ be an acyclic semiautomaton. Then $(Q,A,\cdot)$ is 
 confluent if and only if, 
 for each $q\in Q$, $u,v\in A^*$, we have $q\cdot u\cdot (uv)^{|Q|}=q\cdot v\cdot (uv)^{|Q|}$.
\end{lemma}
\begin{proof}
Assume that  $(Q,A,\cdot)$ is a confluent  acyclic semiautomaton and let 
$q\in Q$, $u,v\in A^*$ be arbitrary. 
We consider the sequence of states 
$$q\cdot u,\ q\cdot u\cdot (uv),\ q\cdot u\cdot (uv)^2,\ \dots,\ 
q\cdot u\cdot (uv)^{|Q|}\,.$$ 
 Since the sequence 
 contains more members than $|Q|$, we have $p=q\cdot u\cdot (uv)^k=q\cdot u\cdot (uv)^\ell$
 for some $0\le k<\ell\le |Q|$. Since $(Q,A,\cdot)$ is acyclic, 
 we have  $p \cdot a =p$ for every $a\in \mathrm{c}(uv)$.
 Therefore, $p=q\cdot u\cdot (uv)^k=q\cdot u\cdot (uv)^{k+1}= \dots =q\cdot u\cdot (uv)^{|Q|}$ 
 and we have $p\cdot w=p$ for every 
 $w\in A^*$ such that $\mathrm{c}(w)\subseteq \mathrm{c}(uv)$. 
 Similarly, for $r=q\cdot v\cdot (uv)^{|Q|}$,  we obtain the same property
 $r\cdot w=r$ for the same words $w$. Taking into account that $(Q,A,\cdot)$ 
 is confluent we obtain the existence of a word $w$ such that 
 $\mathrm{c}(w)\subseteq \mathrm{c}(uv)$ and $p\cdot w=r\cdot w$. 
 Hence $p=r$ and the first implication is proved.
 The second implication is evident.
\qed\end{proof}

Using the condition from Lemma~\ref{l:acyclic-confluent}, one 
can prove that the class of all acyclic confluent semiautomata is a variety of 
semiautomata similarly as in Proposition~\ref{p:counter-free}.
Finally, the main result from~\cite{dlt13-klima} can be formulated in the 
following way. 
It is mentioned in~\cite{dlt13-klima} that the defining condition is testable in a polynomial time.

\begin{proposition}[Klíma and Polák~\cite{dlt13-klima}]
The variety of all acyclic confluent semiautomata  corresponds
to the variety of all piecewise testable languages.
\end{proposition}

\subsection{Ordered Automata with Extensive Actions}\label{s:ord-increasing}

We say that an ordered semiautomaton $(Q,A,\cdot ,\le)$ 
 has \emph{extensive actions} if, for every  
$q\in Q$, $a\in A$, we have  $q\le q\cdot a$.
Clearly, the defining condition is testable in a polynomial time.
The transition ordered monoids of such ordered semiautomata are characterized
by the inequality $1\le x$.
It is known~\cite[Proposition 8.4]{pin-handbook} that the last inequality
characterizes the positive variety of all finite unions of 
languages of the form
$$A^*a_1A^*a_2A^*\dots A^*a_\ell A^* \, , \  \text{ where }\  a_1,\dots ,a_\ell\in A,\
\ell\ge 0\, .$$
Therefore we call them \emph{positive piecewise testable languages}.
In this way one can obtain the following statement, which we prove directly using the theory 
presented in this paper.

\begin{proposition}
The class of all ordered semiautomata with extensive actions 
is a positive variety of ordered semiautomata and  corresponds to
the positive variety of all positive piecewise testable languages.
\end{proposition}
\begin{proof}
It is a routine to check that the class of all ordered semiautomata with extensive actions 
is a positive variety of ordered semiautomata. 
Using Theorem~\ref{t:eilenberg-ordered}, we need to show, that a language $L$ is positive piecewise
testable if and only if its canonical ordered semiautomaton has extensive actions.
To prove that the canonical semiautomaton of a positive piecewise testable language has  extensive
 actions, it is enough to prove this fact 
for languages of the form 
$A^*a_1A^*a_2A^*\dots A^*a_\ell A^*$ with  
$a_1,\dots ,a_\ell\in A,\ \ell\ge 0$. This observation
follows from the description of the canonical ordered
(semi)automata of a language given in Section~\ref{sec:ordered-automata}. 
Indeed, for every language $K=A^*b_1A^*b_2A^*\dots A^*b_k A^*$ we have 
$K\subseteq b^{-1}K$, because $b^{-1}K=K$ or
$b^{-1}K=A^*b_2A^*\dots A^*b_k A^*$ depending on the fact whether 
$b\not=b_1$ or $b=b_1$.

Assume that the canonical automaton 
$\mathcal O_L=(D_L,A,\cdot,\subseteq,L,F_L)$
of a language $L$ has extensive actions;
consequently $(D_L,A,\cdot)$ is an acyclic semiautomaton.
Since $F$ is upward closed,
for every $p\in F$ and $a\in A$, we have $p\cdot a \in F$. 
In other words, for every $p\in F$, we have
$\nL_p=A^*$. 
However by Lemma~\ref{l:canonical-dfa} we have $\nL_p=p$, so 
we get that $F$ contains just one final state $p=A^*$.
Now we consider a simple path in $\mathcal O_L$ from $i$ to $p$ labeled 
by a word $u=a_1a_2\dots a_n$ with $a_k\in A$, 
i.e $i\not= i\cdot a_1\not= i\cdot a_1a_2 \not= \dots \not= i\cdot u=p$.
If we consider a word $w$ such that $w=w_0a_1w_1a_1\dots a_nw_n$, 
where $w_0,w_1,\dots w_n\in A^*$, then
one can easily prove by an induction with respect to $k$ that 
$i\cdot a_1\dots a_k \le i\cdot w_0a_1w_1a_1\dots a_kw_k$.
For $k=n$, we get $p\le i\cdot w\in F$, thus $i\cdot w=p$. Hence 
we can conclude with $A^*a_1A^*a_2 \dots a_n A^* \subseteq L$.
We can consider the language $K$, which is 
the union of such languages $A^*a_1A^*a_2 \dots a_n A^* $ for all possible 
simple paths in $\mathcal O_L$ from $i$ to $p$.  
Now $K\subseteq L$ follows from the previous argument and $L\subseteq K$ is clear, 
because every $w\in L$ describes a unique simple path from $i$ to $p$.
\qed\end{proof}

Note that a usual characterization of the class of positive piecewise testable 
languages is given by a forbidden pattern for DFA (see e.g.~\cite[page 531]{straubing-weil-handbook}).
This pattern consists of two words $v, w\in A^*$ and two  
states $p$ and $q=p\cdot v$ such that $p\cdot w \in F$ and $q\cdot w\not\in F$. 
In view of~(\ref{eq:hopcroft}) from Section~\ref{sec:ordered-automata},
the presence of the pattern is equivalent
to the existence of two states $[p]_\rho \not\le [q]_\rho$, such that $[p]_\rho \cdot_\rho v= [q]_\rho$ 
in the minimal automaton of the language.
The membership for the class of positive piecewise testable languages is decidable in polynomial time --
see~\cite[Corollary 8.5]{pin-handbook} or~\cite[Theorem 2.20]{straubing-weil-handbook}.

\subsection{Autonomous Automata}
We recall examples from the paper~\cite{esik-ito}.
We call a semiautomaton $(Q,A,\cdot)$ {\em autonomous}
if for each state $q\in Q$ and every
pair of letters $a,b\in A$, we have $q\cdot a = q \cdot b$.
For a positive integer $d$, let $\mathbb{V}_d$ be the class of all autonomous semiautomata
being disjoint unions of cycles whose lengths divide $d$. 
Clearly, the defining conditions are testable in a linear time.

\begin{proposition}[\'Esik and Ito~\cite{esik-ito}]
(i) All autonomous semiautomata form a $\mathcal C_l$\=/variety of semiautomata and
the corresponding $\mathcal C_l$\=/variety of languages consists of regular languages
$L$ such that, for all $u,v\in A^*$, if $u\in L$,  
$|u|=|v|$ then $v\in L$.

(ii) The class $\mathbb{V}_d$ 
forms a $\mathcal C_l$\=/variety of semiautomata and
the corresponding $\mathcal C_l$\=/variety of languages consists of all unions
of $(A^d)^*A^i$, $i\in \{0,\dots ,d-1\}$.  
\end{proposition}

\subsection{Synchronizing and Weakly Confluent Automata}
\label{subs:synchro}  
 
Synchronizing automata are 
intensively studied in the literature. 
A semiautomaton $(Q,A,\cdot)$ is {\em synchronizing} if there is a word
$w\in A^*$ such that the set $Q\cdot w$ is a one\=/element set. 
We use an equivalent condition, namely, 
for each pair of states $p,q\in Q$, there exists a word $w\in A^*$ 
such that $p\cdot w=q\cdot w$ (see e.g. Volkov~\cite[Proposition 1]{volkov}).
In this paper we consider the classes of semiautomata
which are closed for taking disjoint unions. So, we need to study 
disjoint unions of synchronizing semiautomata. 
Those automata can be equivalently characterized by the following 
weaker version of confluence. 
We say that a semiautomaton $(Q,A,\cdot)$ is {\em weakly confluent}
if, for each state $q\in Q$ and every
pair of words $u,v\in A^*$,
there is a word $w\in A^*$ such that $(q\cdot u)\cdot w=(q\cdot v)\cdot w$.

\begin{proposition}\label{l:synchronizing}
A semiautomaton is weakly confluent if and only if it is a disjoint 
union of synchronizing semiautomata.
\end{proposition}
\begin{proof}
It is clear that a disjoint union of synchronizing semiautomata is weakly
confluent. 
 
To prove the opposite implication, assume that $(Q,A,\cdot)$ is a weakly confluent semiautomaton.
We consider one connected component and an arbitrary pair $p,q$ of its states. Then there exist
states $p_1, p_2, \dots , p_n$ and letters $a_1, \dots , a_{n-1}$
such that $p_1=p$, $p_n=q$ and for each $i\in\{1, \dots , n-1 \}$
we have $p_i\cdot a_i=p_{i+1}$ or $p_{i+1}\cdot a_i=p_{i}$.  
We claim, for each $i\in\{1,\dots , n\}$, 
the existence of a word $w_i$ such that $p_1\cdot w_i= p_i\cdot w_i$.
This claim gives, in the case $i=n$, that $p\cdot w_n=q\cdot w_n$, which concludes the proof. 
In the rest of the proof we show the claim by the induction 
on $i$.
For $i=1$ one can take any word for $w_1$. Now, assume that the claim is true for $i$, i.e. 
there is a word $w_i$ and state $r_1$ such that 
$r_1=p_1\cdot w_i= p_i\cdot w_i$. Furthermore, we denote $r_2=p_{i+1}\cdot w_i$.
In the case $p_i\cdot a_i=p_{i+1}$, we denote $r_0=p_i$ and we have $r_0 \cdot w_i =r_1$ 
and $r_0\cdot a_iw_i=r_2$. In the case $p_{i+1}\cdot a_i=p_{i}$, we denote
$r_0=p_{i+1}$ and we have  $r_0\cdot a_iw_i=r_1$ and $r_0 \cdot w_i =r_2$. In both cases, 
since the semiautomaton is weakly confluent there exists $u\in A^*$ such that 
$r_1\cdot u=r_2\cdot u$.
Now for $w_{i+1}=w_iu$ we have $p_1\cdot w_{i+1}=(p_1\cdot w_i)\cdot u=r_1\cdot u=r_2\cdot u=
(p_{i+1}\cdot w_i)\cdot u=p_{i+1} \cdot w_{i+1}$. 
\qed\end{proof}

Since the synchronization property can be tested in the polynomial time (see~\cite{volkov}), 
Proposition~\ref{l:synchronizing} 
implies that the weak confluence of a semiautomaton can be tested  in the polynomial 
time, as well.

In the next result we use the category $\mathcal C_s$ of all surjective homomorphisms.
Note that  $f: B^* \rightarrow A^*$ is a surjective homomorphism if and only if 
$A\subseteq f(B)$.
 
\begin{proposition}
The class of all weakly confluent semiautomata is a $\mathcal C_s$\=/variety of semiautomata.
\end{proposition}
\begin{proof}
Clearly, the class of all weakly confluent semiautomata $\mathbb V$ 
is closed under disjoint unions, subsemiautomata and homomorphic images.
 We need to check that $\mathbb V$ is closed under 
direct products of non\=/empty finite families. 

Let $\mathcal Q=(Q,A,\cdot,\le)$ and $\mathcal P=(P,A,\circ,\preceq)$ be a pair of weakly confluent semiautomata.
Take a state $(q,p)$ in the product $\mathcal Q\times \mathcal P$ and let $u,v\in A^*$ be words.
Since $(Q,A,\cdot,\le)$ is weakly confluent, there is $w\in A^*$ such that $q\cdot u\cdot w=q\cdot v\cdot w$.
Now we consider the words $uw$ and $vw$. Since  $\mathcal P$ is weakly confluent, there is $z\in A^*$ such that 
$p\cdot uw \cdot z=p\cdot vw \cdot z$. Hence $(q,p)\cdot u \cdot wz =  (q,p)\cdot v \cdot wz$ and we proved that
$\mathcal Q\times \mathcal P$ is weakly confluent. The general case for a direct product of a non\=/empty finite family
of ordered semiautomata can be proved in the same way.

To finish the proof, assume that $f\in \mathcal C_s(B^*,A^*)$
is a surjective homomorphism.  Let $\mathcal A = (Q,A,\cdot)$ be a weakly confluent semiautomaton
and $\mathcal A^f=(Q,B,\cdot^f)$ is its $f$\=/renaming.
Taking $q\in Q$ and $u,v\in B^*$, we have $q\cdot^f u =q\cdot f(u)$ and  
$q\cdot^f v=q\cdot f(v)$. Since $\mathcal A = (Q,A,\cdot)$ is weakly confluent, there is 
$w\in A^*$ such that $q\cdot f(u) \cdot w= q\cdot f(v) \cdot w$. 
Now we can consider a preimige $w'\in B^*$ of the word $w$ in the surjective homomorphism $f$.
Finally, we can conclude that $(q\cdot^f u) \cdot^f w'=(q\cdot^f v)\cdot^f w'$. 
\qed\end{proof}

\subsection{Automata for Finite Languages}

Finite languages do not form a variety, because their complements, 
the so\=/called {\em cofinite languages}, are not finite. 
Moreover, the class of all finite languages is not closed for taking preimages under 
all homomorphisms.
However, one can restrict the category of homomorphisms to the so\=/called {\em non\=/erasing} ones: 
we say that a homomorphism $f: B^* \rightarrow A^*$ is  
{\em non\=/erasing} if $f^{-1} (\lambda)=\{\lambda\}$. 
The class of all non\=/erasing homomorphisms
is denoted by $\mathcal{C}_{ne}$.
Note that  $\mathcal{C}_{ne}$\=/varieties of languages correspond
to $+$\=/varieties of languages (see~\cite{straubing}).

We use certain technical terminology for states of a given semiautomaton $(Q,A,\cdot)$:
we say that a state $q\in Q$ {\em has a cycle}, if there is a word $u\in A^+$ such that $q\cdot u=q$
and we say that the state $q$ is {\em absorbing} if for each letter $a\in A$ we have $q\cdot a = q$. 

\begin{definition}
We call a semiautomaton $(Q,A,\cdot)$ {\em strongly acyclic}, 
if each state which has a cycle is absorbing.
\end{definition}

It is evident that every strongly acyclic semiautomaton is acyclic.

\begin{proposition}

(i) The class of all strongly acyclic semiautomata forms a $\mathcal{C}_{ne}$\=/variety.

(ii) The class of  all strongly acyclic confluent semiautomata 
forms a $\mathcal{C}_{ne}$\=/variety.
\end{proposition}
\begin{proof}
(i) It is easy to see that the class $\mathbb V$
of all strongly acyclic semiautomata is closed under finite products, 
disjoint unions and subsemiautomata. Also the property of  $f$\=/renaming is clear
whenever we consider 
a non\=/erasing homomorphism $f: B^* \rightarrow A^*$.
Finally, one can prove that the class $\mathbb V$
is closed under homomorphisms in a similar way as in the case of counter\=/free semiautomata.

(ii) By the first part we know that all strongly acyclic semiautomata form a $\mathcal{C}_{ne}$\=/variety.
We also know that all acyclic confluent semiautomata form a variety of semiautomata, and hence they 
form also a $\mathcal{C}_{ne}$\=/variety of semiautomata.
Therefore all strongly acyclic confluent semiautomata, as an intersection 
of two $\mathcal{C}_{ne}$\=/varieties, form a $\mathcal{C}_{ne}$\=/variety again.
\qed\end{proof}

\begin{proposition}\label{p:co-finite-languages}
The  $\mathcal{C}_{ne}$\=/variety of all finite and all cofinite languages  
corresponds to the $\mathcal{C}_{ne}$\=/variety of all strongly acyclic confluent semiautomata.
\end{proposition}
\begin{proof}
At first, consider an arbitrary finite language $L\subseteq A^*$ and its canonical automaton
$(D_L,A,\cdot,L,F_L)$. Since $L$ is finite, 
there is only one  state in $D_L$  which has a cycle, namely the state $\emptyset$. 
Moreover, this state is absorbing and it is reachable from all other states, 
because quotients of finite languages are finite.
Therefore the semiautomaton $(D_L,A,\cdot)$ is 
strongly acyclic and confluent at the same time.
Of course, if we start with the complement of a finite language $L$, 
the canonical semiautomaton is the same as for $L$. 

Conversely,
let $\mathcal A=(Q,A,\cdot)$ be a strongly acyclic confluent semiautomaton. 
For an arbitrary state $q\in Q$, we take 
some path starting in $q$ of length $|Q|$.
On that path there is a state $q'$ which has a cycle, i.e.
$q\cdot u=q'=q'\cdot v$ for some $u\in A^*$, $v\in A^+$.
Since $\mathcal A$ is strongly acyclic, $q'$ is an absorbing state.
Since $\mathcal A$ is confluent, there is at most one such absorbing state $q'$ reachable from $q$.
Now we choose $i\in Q$ and $F\subseteq Q$ arbitrarily and we consider
the automaton $(Q,A,\cdot,i,F)$.
By the previous considerations there is just one state reachable from $i$ which has 
a cycle.  
We denote it by $f$. Note that it is an absorbing state.
One can see that, for each state $q\not= f$, the set
$\{\, u \mid i\cdot u =q\, \}$ is finite and therefore $\{ u \mid i\cdot u=f\}$ is a complement of the finite 
language. Thus depending on the fact $f\in F$, the language recognized by 
$(Q,A,\cdot,i,F)$ is cofinite or finite.
\qed\end{proof}

Naturally, one can try to describe the 
corresponding $\mathcal{C}_{ne}$\=/variety of 
languages for the $\mathcal{C}_{ne}$\=/variety of  
strongly acyclic semiautomata.   
Following Pin~\cite[Section 5.3]{pin-handbook}, we call
$L\subseteq A^*$ a {\em prefix\=/testable} language if $L$ is a finite union of a finite language
and languages of the form $uA^*$, with $u\in A^*$. One can prove
the following statement in a similar way as Proposition~\ref{p:co-finite-languages}. Note that 
one can find also a characterization via syntactic semigroups
in~\cite[Section 5.3]{pin-handbook}.
\begin{proposition}
The $\mathcal{C}_{ne}$\=/variety of all  prefix\=/testable  languages
corresponds to the $\mathcal{C}_{ne}$\=/variety of all strongly acyclic semiautomata.
\end{proposition}

The characterization from Proposition~\ref{p:co-finite-languages} 
can be modified for a positive $\mathcal{C}_{ne}$\=/variety of finite languages
$\mathcal F$: where $\mathcal F(A)$ consists from  $A^*$ and all finite 
languages over $A$. 
To make the 
characterizing condition
more readable, for a given strongly acyclic confluent 
semiautomaton and its state $q$,
we call the uniquely determined state $q'$, mentioned in the proof 
 of Proposition~\ref{p:co-finite-languages},
as a {\em main follower} of the state $q$. 

\begin{proposition}
The positive $\mathcal{C}_{ne}$\=/variety of  all finite languages corresponds to
the positive $\mathcal{C}_{ne}$\=/variety of all strongly acyclic confluent ordered semiautomata
satisfying $q'\le q$ for each state $q$ and its main follower $q'$.
\end{proposition}
\begin{proof}
By the first paragraph of the proof of Proposition~\ref{p:co-finite-languages}, every 
canonical ordered automaton of a finite language satisfies the additional condition 
$q'\le q$ for each state $q$, because the main follower of $q$ is $\emptyset$.

Similarly, in the second part of the proof: 
Let $f$ be the considered main follower of $i$. Since it is also main follower of all reachable 
states from the initial state $i$, we see that $f$ is the minimal state among all reachable states from $i$.
Now if $f$ is final, then all states are final, because the final states form upward closed subset.
Consequently the language accepted by the ordered automaton is $A^*$ in this case. 
If $f$ is not final, then
the language accepted by the ordered automaton is finite.
\qed\end{proof}

Note that, all considered conditions on semiautomata discussed in this subsection 
can be checked in polynomial time.

\subsection{Automata for Languages Closed under Inserting Segments} 
\label{subs:inserting} 

We know that a language $L\subseteq A^*$ is positive piecewise testable if, for every pair of 
words $u,w\in A^*$ such that $uw\in L$ and for a letter $a\in A$, we have $uaw\in L$.
So, we can add an arbitrary letter into each word from the language (at an arbitrary position) 
and the resulting word stays in the language.
Now we consider an analogue, where we put into the word not only a letter but a word of 
a given fixed length. 
The length of a word $v\in A^*$ is denoted by $|v|$ as usually.

For each positive integer 
$n$, we consider the following property of a given regular language $L\subseteq A^*$:
$$\text{ for every } u,v,w \in A^*, \text{ if }  uw\in L
\text{ and } |v|=n,  \text{ then }  uvw \in L \, . $$
We say that $L$ is closed under $n$\=/{\em insertions} whenever $L$ satisfies this property.
We show that the class of all regular languages  closed under $n$\=/insertions form a positive 
$\mathcal C$\=/variety of languages
by describing the corresponding positive $\mathcal C$\=/variety of ordered semiautomata.

At first, we need to describe an appropriate category of homomorphisms.
Let $\mathcal C_{lm}$ be the category 
consisting of the so\=/called {\em length\=/multiplying} (see~\cite{straubing}) homomorphisms:
$f\in \mathcal C_{lm} (B^*, A^*)$
if there exists a positive integer $k$ such that $|f(b)|=k$ for every $b\in B$.

\begin{definition}
Let $n$ be a positive integer and $\mathcal Q=(Q,A,\cdot,\le)$ be an ordered  semiautomaton. 
We say that  $\mathcal Q$ has $n$\=/extensive actions if, for every $q\in Q$ and $u\in A^*$ such that $|u|=n$,
we have $q\le q\cdot u $.
\end{definition}

Note that ordered semiautomata from Subsection~\ref{s:ord-increasing} are ordered semiautomata which have  
$1$\=/extensive actions. Of course, these ordered semiautomata have  
$n$\=/extensive actions for every $n$. More generally,
if $n$ divides $m$ and  an ordered semiautomaton 
$\mathcal Q$ has $n$\=/extensive actions, then $\mathcal Q$  has $m$\=/extensive actions.

\begin{proposition}
Let $n$ be a positive integer.
The class of all ordered semiautomata which have  
$n$\=/extensive actions form a positive $\mathcal C_{lm}$\=/variety of ordered semiautomata.
The corresponding positive $\mathcal C_{lm}$\=/variety of languages consists
of all regular languages closed under $n$\=/insertions.
\end{proposition}
\begin{proof}
The first part of the statement is easy to show.
To establish the second part,
let $L$ be a regular language over $A$ 
closed  under $n$\=/insertions. 
For $u\in A^*$, we consider the state $K=u^{-1}L$ 
in the canonical ordered semiautomaton of $L$. Now 
we show that
for every $v\in A^*$ such that $|v|=n$, we have $K\subseteq v^{-1}K$.
Indeed, if $w\in K=u^{-1}L$ then $uw\in L$
and since $L$ is closed under $n$\=/insertions we get $uvw\in L$.
Hence $vw\in K=u^{-1}L$, which implies $w\in  v^{-1}K$.
Therefor the canonical ordered semiautomaton of $L$
has $n$\=/extensive actions. 

On contrary, let $L$ be recognized by $\mathcal Q=(Q,A,\cdot,\le,i,F)$
with  $n$\=/extensive actions. For every $u,v,w \in A^*$ such that  $uw\in L$ and  $|v|=n$,
we can consider the state $q=i\cdot u$ in $\mathcal Q$. Since  $\mathcal Q$ has 
$n$\=/extensive actions
we have $q\cdot v \ge q$. Hence $i\cdot uvw =q\cdot vw\ \ge q\cdot w =i\cdot uw \in F$
and we can conclude
that $uvw\in L$. Thus $L$ is closed under $n$\=/insertions.
\qed\end{proof}

For a fixed $n$, it is decidable in polynomial time whether 
a given ordered semiautomaton has $n$\=/extensive actions, because 
the relation $q\le q\cdot u $ has to be checked only for polynomially many words $u$. 


\section{Membership Problem for $\mathcal C$\=/Varieties of Semiautomata}
\label{s:membership}

In the previous section, the membership problem for (positive) $\mathcal C$\=/varieties of semiautomata 
was always solved by an ad hoc argument. Here we discuss 
whether it is possible 
to give a general result in this direction. 
For that purpose, recall that $\omega$\=/identity is a pair of $\omega$\=/terms, which are 
constructed from variables by (repeated) successive application of
concatenation and the unary operation $u \mapsto u^\omega$. 
In a particular monoid, the interpretation of this unary operation 
assigns to each element $s$ its
uniquely determined power which is idempotent. 

In the case of $\mathcal C_{all}$ consisting of all homomorphisms, 
we mention Theorem 2.19 from~\cite{straubing-weil-handbook}
which states the following result: if the corresponding pseudovariety of monoids is 
defined by a finite set of $\omega$\=/identities 
then the membership problem of the corresponding
variety of languages is decidable by a polynomial 
space algorithm in the size of the input automaton.
Thus, Theorem 2.19 slightly extends the case when the pseudovariety of monoids 
is defined by a finite set of identities. 
The algorithm checks the defining $\omega$\=/identities in the syntactic monoid $M_L$
of a language $L$ and uses the basic fact that $M_L$ is the transition monoid 
of the minimal automaton of $L$.
This extension is possible, because
the unary operation $(\phantom{u})^\omega$ can be effectively computed 
from the input automaton.

We should mention that checking a fixed identity in an input semiautomaton
can be done in a better way. 
Such a (NL) algorithm (a folklore algorithm in the theory) 
guesses a pair of finite sequences of states for two sides of a given identity $u=v$ 
which are visited during reading the word $u$ (and $v$ respectively) letter by letter.
These sequences have the same first states and distinct last states.  
Then the algorithm checks whether for each variable, there is  
a transition of the automaton given by a word, which transforms all states
in the sequence in the right way, when every occurrence of the variable is considered.
If, for every used variable, there is such a word, we obtained a counterexample disproving
the identity $u=v$.

Whichever algorithm is used, we  can immediately get the generalization to the case of
positive varieties of languages, because checking inequalities can be done in 
the same manner as checking identities.  
However, we want to use the mentioned algorithms to obtain a corresponding result for
positive $\mathcal C$\=/varieties of ordered semiautomata for the categories
used in this paper.  
For such a result we need the following formal definition.
An $\omega$\=/inequality $u\le v$ \emph{holds in an ordered 
semiautomaton $\mathcal O = (Q,A,\cdot, \le)$
with respect to a category $\mathcal C$} if, 
for every $f\in \mathcal C (X^*,A^*)$ with $X$ being the set of variables
occurring in $uv$, and for every $p\in Q$, we have $p\cdot f(u)\le p\cdot f(v)$.
Here $f(u)$ is equal to $f(u')$,
where $u'$ is a word obtained from $u$ if all occurrences of $\omega$ are replaced by 
an exponent $n$ satisfying the equality $s^\omega=s^n$ in the transition monoid 
of $\mathcal O$ for its arbitrary element $s$.

\begin{theorem}\label{t:membership-problem} 
Let $\mathcal O = (Q,A,\cdot, \le)$ be an ordered semiautomaton, 
let $u\le v$ be an $\omega$\=/inequality and $\mathcal C$ be 
one of the categories $\mathcal C_{ne}$, $\mathcal C_{l}$, $\mathcal C_{s}$ and $\mathcal C_{lm}$. 
The problem whether $u\le v$ holds in $\mathcal O$
with respect to $\mathcal C$ is decidable.
\end{theorem}

\begin{proof}
The result is a consequence of the following propositions.

\begin{proposition}  
Let $\mathcal O = (Q,A,\cdot, \le)$ be an ordered semiautomaton, 
$u\le v$ be an $\omega$\=/inequality and $\mathcal C$ be 
one of the categories $\mathcal C_{ne}$, $\mathcal C_{l}$ and $\mathcal C_{s}$. 
The problem of deciding
whether $u\le v$ holds in $\mathcal O$
with respect to $\mathcal C$ can be solved by a polynomial space algorithm.
\end{proposition}

\begin{proof}
First of all, we prove the statement formally
for $\mathcal C=\mathcal C_{all}$.
We start with the case when $\omega$ operation is not used.
It is mentioned in Section~\ref{s:membership} that such an algorithm is a
folklore in the theory.

Let $y_1 \ldots y_s \le z_1 \ldots z_t$ be an inequality,
where $y_1, \dots, y_s, z_1, \dots, z_t$ are variables from $X$.
Recall, that this inequality holds in the transition ordered monoid of~$\mathcal{A}$
if and only if for every homomorphism $f: X^* \rightarrow A^*$, the inequality of transformations
$f(y_1) \circ \dots \circ f(y_s) \le f(z_1) \circ \dots \circ f(z_t)$
is satisfied.
This requirement can be reformulated as the inequality
$q \cdot f(y_1) \cdots f(y_s) \le q \cdot f(z_1) \cdots f(z_t)$
of states of~$\mathcal{A}$, for every state $q \in Q$
and every homomorphism $f: X^* \rightarrow A^*$.
This means that the inequality is not valid if and only if
there exist such $f$ and states
$p_0, p_1, \dots, p_s, q_0, q_1, \dots, q_t \in Q$,
with $p_0 = q_0$ and $p_s \not\le q_t$, which satisfy
$p_{i-1} \cdot f(y_i) = p_i$ and $q_{j-1} \cdot f(z_j) = q_j$
for every $i \in \{1,\dots,s\}$ and $j \in \{1,\dots,t\}$.
Since the numbers $s$ and $t$ are constants,
one can non\=/deterministically choose all these states,
and then decide whether for this choice of states
the required homomorphism $f$ exists.
For every variable~$x$,
denote by $I_x$ the set of all $i \in \{1,\dots,s\}$
such that $y_i = x$,
and by $J_x$ the set of all $j \in \{1,\dots,t\}$
such that $z_j = x$.
In order to decide existence of~$f$,
one has to check whether for every variable~$x$
there exists a word $f(x) \in A^*$ such that
$p_{i-1} \cdot f(x) = p_i$ for every $i \in I_x$
and $q_{j-1} \cdot f(x) = q_j$ for every $j \in J_x$.
However, the existence of such a word $f(x)$ can be expressed
as a condition on the product automaton of
$|I_x| + |J_x|$ copies of the automaton~$\mathcal{A}$;
namely, it is equivalent to reachability of the state with
components $p_i$, for $i \in I_x$, and $q_j$, for $j \in J_x$,
from the state with components $p_{i-1}$, for $i \in I_x$,
and $q_{j-1}$, for $j \in J_x$.
Recall that $|I_x| + |J_x|$  is a constant.

Now assume that $u$ and $v$ are $\omega$\=/terms.
We are guessing the states as in the previous simple case, but we 
do this inductively with respect to the structure of the 
$\omega$\=/terms $u$ and $v$ from top to down.
In this way we obtain a more complicated system of states comparing the sequences 
in the case of (linear) words.
To explain the inductive construction, assume that  
we have guessed states $p$ and $q$ for a certain $\omega$\=/subterm $w$ 
assuming that $p\cdot f(w)=q$.
If $w=w_1w_2$ for $\omega$\=/terms $w_1$ and $w_2$, then we simply guess a state $r$ and assume
that $p\cdot f(w_1)=r$ and $r\cdot f(w_2)=p$. 
The case $w=z^\omega$, with a subterm $z$, is more complicated. 
It is well known that, for every element $s$ in the transition monoid of the given automaton,   
the element $s^\omega$ is equal to $s^n$ for some  $n\le |Q|$. In particular, $s^n\cdot s^n=s^n$ holds
for this $n$. 
So, we guess $n\le |Q| $ and states $r_0,r_1,r_2,\dots,  r_{2n}$ such that
$r_0=p$, $r_n=r_{2n}=q$ and we assume that $r_{i-1}\cdot f(w)= r_{i}$ for every $i=1,\dots, 2n$.
In this way, when we decompose all subterms, we obtain a system of states equipped with 
assumptions of the form $p\cdot f(x)=q$, where $p$ and $q$ are states and $x$ is a variable.
Since the $\omega$\=/terms $u$ and $v$ are not part of the input, there are only 
constantly many steps of the algorithm decomposing the terms $u$ and $v$. 
Thus, at the end, the number of conditions is polynomial with respect the size of the input automaton. 
(In fact, the number of the conditions 
can be bounded by the number of all states in $Q$, which is linear.)  
The final part of the algorithm  is the same: 
we just check, for each variable $x$, whether it is possible to satisfy 
all the conditions concerning $f(x)$ at the same time.
Point out, that the number of conditions was constant in the case of identity $u=v$ in the first part, 
which gives $\log$ space algorithm in the that case.

Now we are ready to discuss another categories, where we search for 
$f\in \mathcal C(X^*,A^*)$.
The case $\mathcal C=\mathcal C_{ne}$ is trivial. 
When we test reachability in the product of certain number of copies of $\mathcal O$, 
we are looking for a non\=/empty path in the graph.
The case $\mathcal C=\mathcal C_{l}$ is even easier, because we test reachability in one step.
Seeing this case from another point of view,
this case is easy, because there are only polynomially many homomorphisms
in $\mathcal C_{l}(X^*,A^*)$ for fixed $X$ and $A$
where only $A$ is a part of the input.
The case $\mathcal C=\mathcal C_{s}$ is also easy. We are looking for $f\in \mathcal C(X^*,A^*)$
such that $A\subseteq f(X)$. So, we can additionally guess, for each letter $a\in A$, a variable $x\in X$ such that
$f(x)=a$.

We could conclude with the remark, that is well known
that nondeterministic polynomial space
is equivalent to deterministic polynomial space.
\qed
\end{proof}

\begin{proposition}  
Let $\mathcal O = (Q,A,\cdot, \le)$ be an ordered semiautomaton, 
$u\le v$ be an $\omega$\=/inequality. 
The problem 
whether $u\le v$ holds in $\mathcal O$
with respect to $\mathcal C_{lm}$ is decidable.
\end{proposition}

\begin{proof}
We proceed as in the general case up to the place where the existence of $f(x)$ is discussed. 
We do not decide whether there is  $f(x)\in A^*$ satisfying all conditions before
we first complete the conditions in such a way that, for every $q\in Q$, the condition
on $q\cdot f(x)$ is present. This is made by guessing missing pairs $q\cdot f(x)$ for all $q$ and $x$. 
Just now we test whether there are words $f(x)$ satisfying the conditions. 

Only if there are such words, we continue. Next we try to describe all of them.
It is possible, because, for every $x$, we know how $f(x)$ transform the semiautomaton $\mathcal O$.
So, the language of all words which are considered as a potential words $f(x)$ is a regular language
which is recognized by the transition monoid of the semiautomaton $\mathcal O$. 
We denote it as $L_x$.
Furthermore, we are able to compute a regular expression $r_x$ describing $L_x$.  
We need to decide whether for each variable $x$ there is a word $w_x\in L_x$ such that all words's
$w_x$ have the same length. Thus, we need to know all possible lengths of words in $L_x$.
For this purpose we consider the unique literal mapping $\psi : A^*\rightarrow \{a\}^*$, $\psi(A)=\{a\}$.
Clearly, the language $\psi(L_x)$ is regular, because it is described by a regular expression $\ol{r_x}$, 
which  can be obtained from $r_x$, if we replace every letter from the alphabet $A$ by the letter $a$.
Moreover, there is a word $w_x\in L_x$ of length $k$ if and only if $a^k\in \psi(L_x)$.
So, the existence of an integer $k$ such that $L_x\cap A^k\not=\emptyset$ holds for every $x$,
is equivalent to the fact $\bigcap_{x\in X} \psi(L_x) \not= \emptyset$. The later inequality 
is equivalent to non\=/emptiness of the language given by the 
generalized regular expression $\bigcap_{x\in X} \ol{r_x}$.
So, one can decide this question.
\qed
\end{proof}
We did not discuss the complexity of the algorithm, because
we do not see how to effectively construct the regular expression $\bigcap_{x\in X} \ol{r_x}$. 
\end{proof}

\section{Further Remarks}

At the end we could mention that
one can extend the construction in at least two natural directions.
First, the theory of tree languages is a field where many fundamental ideas from the theory of deterministic 
automata were successfully generalized. 
Another recent notion of biautomata (see~\cite{biautomaty} and \cite{ncma13-holzer})
is based on considering both\=/sided quotients instead of left quotients only.
In both cases one can try to apply the previous constructions
and consider varieties of (semi)automata. 
Some papers in this direction already exist~\cite{esik-ivan}.

\end{document}